\documentclass[preprint,10pt]{elsarticle}

\usepackage{mathrsfs, amsthm, amsmath}

\usepackage{bbm}

\usepackage{graphicx}

\usepackage{amssymb}

\newtheorem{theorem}{Theorem}[section]
\newtheorem{lemma}[theorem]{Lemma}
\newtheorem{cor}[theorem]{Corollary}
\newtheorem{prop}[theorem]{Proposition}

\theoremstyle{definition}
\newtheorem{defn}[theorem]{Definition}
\newtheorem{example}[theorem]{Example}

\theoremstyle{remark}
\newtheorem{remark}[theorem]{Remark}

\numberwithin{equation}{section}

\DeclareMathAlphabet      {\mathbfit}{OML}{cmm}{b}{it}
\let\bm=\mathbfit

\let\text=\mbox

\catcode`\@=11

\renewcommand{\a}{\alpha}
\renewcommand{\b}{\beta}

\renewcommand{\d}{\delta}
\newcommand{\g}{\lambda}
\renewcommand{\o}{\omega}
\newcommand{\q}{\quad}

\newcommand{\s}{\sigma}

\newcommand{\M}{{\cal M}}

\newcommand{\ty}{\infty}
\newcommand{\e}{\varepsilon}
\newcommand{\f}{\varphi}
\newcommand{\ov}[1]{\overline{#1}}

\newcommand{\D}{\mathrm d}
\renewcommand{\O}{\Omega}
\newcommand{\pa}{\partial}

\newcommand{\I}{\mathbbm{i}}
\newcommand{\E}{\mathrm e}

\newcommand{\st}{\subset}
\newcommand{\stq}{\subseteq}

\newcount\br@j
\newcommand{\udesno}[1]{\unskip\nobreak\hfil\penalty50\hskip1em\hbox{}
             \nobreak\hfil{#1\unskip\ignorespaces}
                 \parfillskip=\z@ \finalhyphendemerits=\z@\par
                 \parfillskip=0pt plus 1fil}
\catcode`\@=11

\newcommand{\eR}{\mathbb{R}}
\newcommand{\eN}{\mathbb{N}}
\newcommand{\Ze}{\mathbb{Z}}
\newcommand{\Qu}{\mathbb{Q}}
\newcommand{\Ce}{\mathbb{C}}

\newcommand{\re}{\mathop{\mathrm{Re}}}
\newcommand{\im}{\mathop{\mathrm{Im}}}

\newcommand{\po}{{\mathop{\mathcal P}}}

\newcommand{\res}{\operatorname{res}}

\newcommand{\sideremark}[1]{\ifvmode\leavevmode\fi\vadjust{\vbox to0pt{\vss 
      \hbox to 0pt{\hskip\hsize\hskip1em           
 \vbox{\hsize2cm\tiny\raggedright\pretolerance10000
 \noindent #1\hfill}\hss}\vbox to8pt{\vfil}\vss}}}%

\journal{ArXiv}

\begin{document}

\begin{frontmatter}


\title{Distance and tube zeta functions of 
fractals and arbitrary compact sets}%



\author[label1]{M.\ L.\ Lapidus\corref{cor1}}
\address[label1]{University of California, Department of Mathematics, 900 University Ave., Riverside, California 92521-0135,	 USA}
\ead{lapidus@math.ucr.edu}



\author[label2]{G.\ Radunovi\'{c}\corref{cor2}\fnref{cor2a}}

\ead{goran.radunovic@fer.hr}


\author[label3]{D.\ \v Zubrini\'{c}\corref{cor3}}
\ead{darko.zubrinic@fer.hr}

\address[label2,label3]{University of Zagreb, Faculty of Electrical Engineering and Computing, Department of Applied Mathematics, Unska 3, 10000 Zagreb, Croatia}

\fntext[cor2a]{Corresponding author}

\cortext[cor1]{The research of Michel L.~Lapidus was partially supported by the U.S.\ National Science
Foundation under grants ~DMS-0707524 and DMS-1107750, as well as by the Institut des Hautes \' Etudes Scientifiques (IH\' ES) where the first author was a visiting professor in the Spring of 2012 while part of this research was completed.}

\cortext[cor2,cor3]{The research of Goran Radunovi\'c and Darko \v Zubrini\'c was supported in part by the Croatian Science Foundation under the project IP-2014-09-2285 and by the Franco-Croatian 
PHC-COGITO project.}







\begin{abstract}
Recently, the first author has extended the definition of the zeta function associated with fractal
strings to arbitrary bounded subsets $A$ of the $N$-dimensional Euclidean space ${\mathbb R}^N$, for any integer $N\ge1$.
It is defined by the Lebesgue integral $\zeta_A(s)=\int_{A_{\delta}}d(x,A)^{s-N}\D x$, for all $s\in\Ce$ with $\operatorname{Re}\,s$ sufficiently large, and we call it the {\em distance zeta function} of $A$. 
Here, $d(x,A)$\label{d(x,A)} denotes the 
Euclidean distance from $x$ to $A$ and $A_{\delta}$ is the $\delta$-neigh\-bor\-hood of~$A$, where $\d$ is a fixed positive real number.
We prove that the abscissa of absolute convergence of $\zeta_A$ is equal to $\overline\dim_BA$, 
 the upper box (or Minkowski) dimension of~$A$. 
Particular attention is payed to
the principal complex dimensions of $A$, defined as the set of poles of $\zeta_A$ located on
the critical line $\{\mathop{\mathrm{Re}} s=\overline\dim_BA\}$, provided $\zeta_A$ possesses a meromorphic extension to a neighborhood of the critical line. We also introduce a new, closely related zeta function, $\tilde\zeta_A(s)=\int_0^\d t^{s-N-1}|A_t|\,\D t$, called the {\em tube zeta function} of $A$. Assuming that $A$ is Minkowski measurable, we show that, under some mild conditions,
the residue of $\tilde\zeta_A$ computed at $D=\dim_BA$ (the box dimension of $A$), is equal to
 the Minkowski content of $A$. More generally, without assuming that $A$ is Minkowski measurable, we show that the residue is squeezed between the lower and upper Minkowski contents of $A$. 
 We also introduce {\em transcendentally quasiperiodic sets}, and construct a class of such sets, using generalized Cantor sets, along with Baker's theorem from
the theory of transcendental numbers. 
\end{abstract}


\end{frontmatter}

\tableofcontents

\section{Introduction}\label{intro}




 In this article, we provide a far-reaching extension of the theory of zeta functions for fractal strings,
to arbitrary fractal sets 
in Euclidean spaces of any dimension. Fractal strings have been introduced by the first author (M.\ L.\ Lapidus) in the early 1990s. The related theory of zeta functions of fractal strings and their complex dimensions, developed in the course of the last two decades of
active research, is presented in an extensive monograph
of the first author with M.\ van Frankenhuijsen \cite{lapidusfrank12}. 

The new zeta function $\zeta_A$, associated with any fractal set $A$ in $\eR^N$, has been introduced
in 2009 by the first author, and its definition can be found in Equation (\ref{z}) below. We refer to it as the {\em distance zeta function} of $A$. 
Here, by a fractal set,
we mean any  bounded set $A$ of the Euclidean space $\eR^N$, with $N\geq 1$. The reason is that, in this paper, the 
key role is played by a certain notion of fractal dimension, more specifically, by the upper
box dimension of a bounded set (also called the upper Minkowski dimension, Bouligand dimension, or limit capacity, etc.). This new class of zeta functions enables us to obtain a
nontrivial extension of the theory of {\em complex dimensions of fractal strings}, to arbitrary bounded fractal sets in Euclidean spaces of any dimension.

A systematic study of the zeta functions associated with fractal strings and fractal sprays was motivated and undertaken, in particular, in the 1990s in papers of the first author, [Lap1--3], as well as in joint papers of the first author with C.~Pomerance
[{LapPo1--2}] and with H.\ Maier \cite{LapMa2}. 
In a series of papers, as well as in two monographs with M.\ van Frankenhuijsen [{Lap-vFr1--2}], 
and in the book \cite{lapz}, it has grown into a well-established theory of fractal complex dimensions, and is still an active area of research, with applications to a variety of subjects, including spectral theory, harmonic analysis, number theory, dynamical systems, probability theory and mathematical physics. We also draw the reader's attention to 
\cite{DubSep},  
[Es1--2], 
[EsLi1--2], 
\cite{fal2},
\cite{hamlap}, 
\cite{lapidushe}, 
\cite{HerLa1}, 
\cite{Kom},  
\cite{LaLeRo}, 
\cite{LaLu}, 
\cite{lappe2}, 
\cite{lappewi1}, 
[LapRa\v Zu1--8], 
\cite{lemen}, 
\cite{MorSep}, 
[MorSepVi1--2], 
[Ol1--2],
\cite{winter}, 
[Tep1--2], 
along with the many relevant references therein.

Other, very different  approaches to a higher-dimensional theory of some special classes of fractal sets, namely, fractal sprays and self-similar tilings, were developed by the first author and E.\ Pearse in \cite{lappe2}, as well as by the first author, E.~Pearse and S.~Winter
in \cite{lappewi1} 
via fractal tube formulas and the associated scaling and tubular zeta functions. 
(See also \cite{pe2} and \cite{pewi}.)

The definitions of the tubular zeta functions introduced in \cite{lappe2}
and \cite{lappewi1} differ considerably from those studied in this article. 
The precise connection between these zeta functions and the fractal zeta functions introduced in this paper is provided in \cite{cras2}.
We point out that by using the fractal zeta functions introduced in this paper, it is possible to generalize the fractal tube formulas and a Minkowski measurability criterion obtained for fractal strings in \cite{lapidusfrank12} to arbitrary compact sets in Euclidean spaces; see [LapRa\v Zu4--5]. We also refer to \cite{mm} for a key use of these fractal tube formulas to obtain a Minkowski measurability criterion, expressed in terms of the nonexistence of nonreal (principal) complex dimensions and generalizing to any dimension its counterpart established for fractal strings in \cite[Chapter~8]{lapidusfrank12}.
\medskip

\subsection{Contents}\label{contents}
The rest of this paper is organized as follows:

In  Section \ref{ch_distance}, the distance zeta function $\zeta_A$ of a bounded set $A\st\eR^N$ is introduced in Definition \ref{z}. Then, the main result of Section \ref{ch_distance} is obtained in Theorem \ref{an}, in which it is shown (among other things) that the abscissa of (absolute) convergence of the distance zeta function $\zeta_A$ of any bounded subset $A$ of $\eR^N$ is equal to $\ov\dim_BA$, the upper box dimension (or the upper Minkowski dimension) of $A$. (All of the subsets  denoted by $A$ appearing in this paper are implicitly assumed to be nonempty.) As a useful technical tool in the study of fractal zeta functions, we introduce the notion of `equivalence' between tamed Dirichlet-type integrals (see Definition \ref{equ}). We also define the set of `principal complex dimensions' of $A$, denoted by $\dim_{PC}A$ (see Definition \ref{dimc}), as a refinement of the notion of the upper box dimension of $A$. 
Moreover, in the one-dimensional case (i.e., in the case of a bounded fractal string $\mathcal L$), we show that $\zeta_A$, the distance zeta function of $A$ (the boundary of the string $\mathcal L$), and $\zeta_{\mathcal L}$, the geometric zeta function of $\mathcal L$, contain essentially the same information. In particular, $\zeta_A$ and $\zeta_{\mathcal L}$ are equivalent in the above sense, and hence, have the same principal complex dimensions (see Subsections \ref{zeta_s} and \ref{eqzf});
they also have the same (visible) complex dimensions (with the same multiplicities) in every domain of $\Ce\setminus\{0\}$ to which one (and hence, both) of these zeta functions can be meromorphically continued.
Finally, we show that the distance zeta function has a nice and very useful scaling property; see Proposition \ref{scalingd}.

In Section \ref{residues_m}, we introduce the so-called `tube zeta function' $\tilde\zeta_A$ of the bounded set $A$ (which is closely related to the distance zeta function $\zeta_A$; see Theorem \ref{equr} and the associated functional equation \eqref{equ_tilde}), and study its properties; see, in particular, Definition \ref{zeta_tilde} in Subsection \ref{residues_m_tube}. Under suitable natural conditions, we show that the residue of the tube zeta function $\tilde\zeta_A$, computed at $D=\dim_BA$ (assuming that the box dimensions exists), always lies between the lower and upper ($D$-dimensional) Minkowski contents of $A$; see Theorem \ref{pole1mink_tilde}. In particular, if $A$ is Minkowski measurable, then the residue of $\tilde\zeta_A$ at $D$ coincides with the Minkowski content of $A$. Similar results are obtained for the distance zeta function $\zeta_A$ of the fractal set $A$; see Theorem \ref{pole1}. 
In fact, we also show that $\zeta_A$ and $\tilde\zeta_A$, the distance and tube zeta functions of $A$, contain essentially the same information.
These results are illustrated by means of several examples, including
a class of generalized Cantor sets (Examples \ref{res-cantor}, \ref{res-cantor2} and \ref{Cmae}), the $(N-1)$-dimensional sphere in $\eR^N$ (see Example~\ref{sphere}),  $a$-strings (Example \ref{a-string2}), as well as `fractal grills' introduced in Subsection \ref{dtx}; see Theorem~\ref{Axm}.  

In Section \ref{quasi0}, we introduce a class of `$n$-quasiperiodic sets' (Definition \ref{quasiperiodic}). The main result is stated in Theorem \ref{quasi1}, which can be considered as a fractal set-theoretic interpretation of Baker's theorem (Theorem \ref{baker0}) from transcendental number theory
and in which we construct a family of transcendentally $n$-quasiperiodic sets, for any integer $n\ge2$.
An important role in the construction of quasiperiodic sets is played by the class of generalized Cantor sets $C^{(m,a)}$ depending on two parameters, introduced in Definition \ref{Cma}.
Moreover, in Subsection \ref{hyperfractal}, we close the main part of this paper by connecting the present work to future extensions (notably, the construction of transcendentally $\ty$-quasiperiodic sets), the notion of hyperfractal (and even, maximally hyperfractal) set, and more broadly, the notion of fractality within the context of this new general theory of complex dimensions. In short, much as in \cite{lapidusfrank12}, we say that a bounded subset $A\st\eR^N$ is {\em fractal} if its associated zeta function (i.e., the distance or the tube zeta function, $\zeta_A$ or $\tilde\zeta_A$, of $A$ or when $N=1$, the geometric zeta function $\zeta_{\mathcal L}$, where $\mathcal L$ is the fractal string associated with $A$) has at least one nonreal complex dimension or else has a natural boundary beyond which it cannot be meromorphically continued (i.e., $A$ is ``hyperfractal'').
Observe that, unlike in the one-dimensional theory of complex dimensions developed in \cite{lapidusfrank12}, we now have at our disposal precise definitions of the fractal zeta functions of arbitrary bounded subsets of $\eR^N$ and hence, of the complex dimensions of those sets (i.e., of the poles of these fractal zeta functions); see Definition \ref{1.331/2} and the beginning of Subsection \ref{residues_m_tube}. The complex dimensions of a variety of classic and less well-known fractals will be computed in subsequent work, [LapRa\v Zu1--8]. 

The aim of Appendix A is to introduce the class of `extended Dirichlet-type integrals' (or functions), i.e., of EDTIs, which contains all of the fractal zeta functions studied in the present paper; see Definition A.1.  We study some of the key properties of EDTIs and introduce two closely related (but distinct) notions of equivalence; see Definitions A.2 and~A.6.

\bigskip

\subsection{Notation}\label{notation} 
Throughout this paper, we shall use the following notation.
By $|E|=|E|_N$,
 we denote the $N$-dimensional Lebesgue measure of a measurable 
subset $E$ of $\eR^N$.
Given $r\ge0$, the {\em lower and upper $r$-dimensional Minkowski contents} $\M^{*r}(A)$ and $\M_*^r(A)$ of a bounded subset $A$ of $\eR^N$ are defined by 
\begin{equation}\label{mink}
\M_*^r(A)=\liminf_{t\to 0^+}\frac{|A_t|}{t^{N-r}}, \q \M^{*r}(A)=\limsup_{t\to 0^+}\frac{|A_t|}{t^{N-r}}.
\end{equation}
Here,  $A_t:=\{x\in\eR^N:d(x,A)<t\}$ denotes the $t$-{\em neighborhood} (or {\em tubular neighborhood of radius} $t$) of $A$, and $d(x,A)$ is the
Euclidean distance from $x$ to $A$.
The function $t\mapsto|A_t|$, defined for $t$ positive and close to $0$, is called the {\em tube function} associated with $A$.
From our point of view, one of the basic tasks of fractal analysis is to understand the nature of the tube functions for various fractal sets.
The above definition coincides with Federer's definition (in \cite{federer}),
 up to a (positive) multiplicative constant depending only on $N$ and $r$, the value of which is not important for the purposes of this article.

The {\em upper box dimension} of $A$ is defined by
\begin{equation}\label{dim}
\ov\dim_BA=\inf\{r\ge0:\M^{*r}(A)=0\};
\end{equation}
it is easy to see that we also have
\begin{equation}\label{Mty}
\ov\dim_BA=\sup\{r\ge0:\M^{*r}(A)=+\ty\}.
\end{equation}
The lower box dimension of $A$, denoted by $\underline \dim_BA$,  is defined analogously, with $\M_*^r(A)$ instead of  $\M^{*r}(A)$ on the right-hand side of  (\ref{dim}) and (\ref{Mty}). 
Clearly, since $A$ is bounded, we always have $0\le\underline\dim_BA\le\ov\dim_BA\le N$.
If both 
$\ov\dim_BA$ and $\underline\dim_BA$
coincide, their common value is denoted by $\dim_BA$ and is called the {\em box dimension} of $A$ (or {\em Minkowski--Bouligand dimension}, or else,  {\em Minkowski dimension}\label{MinkDim}).
Various properties of the box dimension can be found, e.g., in \cite{falc}, \cite{mattila}, \cite{tricot} and \cite{lapidusfrank12}.

If there exists a nonnegative real number $D$ such that 
$$
0<\M_*^D(A)\le\M^{*D}(A)<\ty,
$$
we say that $A$ is {\em Minkowski nondegenerate}.
If $A$ is nondegenerate, it then follows that $\dim_BA$ exists and is equal to~$D$.
 If $\M_*^D(A)=\M^{*D}(A)$, their common value is denoted by $\M^D(A)$ and called the {\em Minkowski content} of $A$.
 If, in addition, 
$$
\M^D(A)\in(0,+\ty),
$$ 
then $A$ is said to be {\em Minkowski measurable}.\label{Minkowski_measurable}
The notion of Minkowski measurability seems to have been introduced by Hadwiger in \cite{hadwiger} and was later used by Federer in \cite{federer}, as well as by  Stach\'o \cite{stacho} (inspired by \cite{federer}), and in many other works, including 
\cite{BroCar}, 
\cite{Lap1} 
\cite{lapiduspom},  
\cite{fal2},
\cite{tricot},  
\cite{mink}, 
\cite{rae}, 
\cite{Kom} and
\cite{KomPeWi}. 
The notion of Minkowski nondegeneracy has been introduced in \cite{rae} (and was studied earlier in \cite{lapiduspom} and [Lap-vFr1--2] when $N=1$; see also \cite{LapPo2} when $N\ge3$). The notion of Minkowski (or box) dimension was introduced by Bouligand in~\cite{bouligand}. Throughout this paper, we will assume implicitly that the bounded set $A\st\eR^N$ is nonempty.



\medskip

We note that since $|A_t|=|(\ov A)_t|$ for every $t>0$, the values of $\M_*^r(A)$, $\M^{*r}(A)$, $\underline\dim_BA$, $\ov\dim_BA$ (as well as of $\M^D(A)$ and $\dim_BA$, when they exist) do not change when we replace the bounded set $A\st\eR^N$ by its closure $\ov A$ in $\eR^N$.
Therefore, throughout this paper, we might as well assume a priori that $A$ is an arbitrary (nonempty) compact subset of $\eR^N$.
Observe that, as is well known, this is in sharp contrast with the Hausdorff dimension (and associated Hausdorff measure $\mathcal{H}_H$); see, e.g., \cite{falc}. For example, if $A=\{1/j:j\in\eN\}$, then (since $A$ is countable), $\dim_HA=0$ and $\mathcal{H}_H(A)=0$, while $D:=\dim_BA=1/2$ and $\mathcal{M}^D(A)=2\sqrt2$; see \cite[Example 5.1]{Lap1}.


Finally, given an extended real number $\a\in\eR\cup\{\pm\ty\}$, we denote by $\{\re s>\a\}$ the open right half-plane $\{s\in\Ce:\re s>\a\}$
(which coincides with  $\Ce$ or $\emptyset$ if $\a=-\ty$ or $+\ty$, respectively). Furthermore, if $\a\in\eR$, we denote by $\{\re s=\a\}$ the vertical line $\{s\in\Ce:\re s=\a\}$. Also, we let $\I:=\sqrt{-1}$.

\section{Distance and tube zeta functions of fractal sets}\label{ch_distance}

In this section, we introduce and study a new fractal zeta function, namely, the distance zeta function attached to an arbitrary bounded subset of $\eR^N$, for any $N\ge1$; see Subection~\ref{properties_definition} and Subsection~\ref{properties_analyticity}. In Subsection~\ref{properties_zeta}, we then consider the special case when $N=1$ and compare this new fractal zeta function with the known geometric zeta function of a fractal string. Finally, in Subsection \ref{eqzf}, we introduce a suitable equivalence relation which enables us to capture some of the main features of fractal zeta functions.



\subsection{Definition of the distance zeta functions of fractal sets}\label{properties_definition} 
We study here some basic properties of the distance zeta function
$\zeta_A=\zeta_A(s)$ associated with an arbitrary bounded subset $A$ of $\eR^N$, and introduced by the first author in 2009. 

\begin{defn}\label{defn}  Let $\delta$ be any given positive number.
The {\em distance zeta function} $\zeta_A$ of a bounded subset $A$ of $\eR^N$ is defined by
\begin{equation}\label{z}
\zeta_A(s):=\int_{A_\delta}d(x,A)^{s-N}\D x.
\end{equation}
Here, the integral is taken in the sense of Lebesgue (hence, the complex-valued function $d(\,\cdot\,,A)^{s-N}$ is absolutely integrable on $A_\d$) and we assume that $s\in\Ce$ is such that $\re s$ is sufficiently large.
\end{defn}

As we shall see in Theorem \ref{an}, 
the Lebesgue integral in (\ref{z}) is well defined if $\re s$ is larger than $\ov\dim_BA$, the upper box dimension of $A$; furthermore, $\ov\dim_BA=D(\zeta_A)$, the {\em abscissa of $($absolute$)$ convergence} of $\zeta_A$. Moreover,
under the additional hypotheses of Theorem \ref{an}(c), $\ov\dim_BA$ also coincides with $D_{\rm hol}(\zeta_A)$, the {\em abscissa of holomorphic continuation} of $\zeta_A$. Here, by definition,
\begin{equation}\label{DzetaA}
D(\zeta_A):=\inf\left\{\a\in\eR:\int_{A_\d}d(x,A)^{\a-N}\D x<\ty\right\}
\end{equation}
while
\begin{equation}\label{Dhol}
D_{\rm hol}(\zeta_A):=\inf\big\{\a\in\eR:\mbox{$\zeta_A$ is holomorphic on $\{\re s>\a\}$}\,\big\}.
\end{equation}
Hence, the {\em half-plane of $($absolute$)$ convergence} of $\zeta_A$, $\Pi(\zeta_A):=\{\re s>D(\zeta_A)\}$ (resp., the {\em half-plane of holomorphic continuation} of $\zeta_A$, $\mathcal{H}(\zeta_A):=\{\re s>D_{\rm hol}(\zeta_A)\}$)
is the largest open half-plane of the form $\{\re s>\a\}$, for some $\a\in\eR\cup\{\pm\ty\}$, on which the Lebesgue integral $\int_{A_\d}d(x,A)^{s-N}\D x$ is convergent or, equivalently, absolutely convergent (resp., to which $\zeta_A$ can be holomorphically continued). It will follow from our results that $D(\zeta_A)\in[0,N]$ while $D_{\rm hol}(\zeta_A)\in[-\ty,D(\zeta_A)]$, and that both $D(\zeta_A)$ and $D_{\rm hol}(\zeta_A)$ are independent of the choice of $\d>0$; 
see Proposition \ref{O} along with Definition \ref{equ}.

Again, the same comment can be made about $D(\tilde\zeta_A)$ and $D_{\rm hol}(\tilde\zeta_A)$, given exactly as in \eqref{DzetaA} and \eqref{Dhol}, respectively, except for $\zeta_A$ replaced by $\tilde\zeta_A$ (the tube zeta function of $A$, see Definition \ref{zeta_tilde}). Actually, if $\ov\dim_BA<N$, then $D(\zeta_A)=D(\tilde\zeta_A)$ and $D_{\rm hol}(\zeta_A)=D_{\rm hol}(\tilde\zeta_A)$; see Corollary \ref{equr_c}.

Given any meromorphic function $f$, the {\em abscissa of holomorphic continuation} of $f$, denoted by $D_{\rm hol}(f)$, can be defined in exactly the same way as $D_{\rm hol}(\zeta_A)$, except with $\zeta_A$ replaced by $f$ in the counterpart of \eqref{Dhol}. The same comment is not true for $D(f)$, which may not make sense unless $f$ is given by a Dirichlet-type integral (DTI); see  Subsection \ref{eqzf} and Appendix A below.

As will be shown in Proposition \ref{O}, the dependence of $\zeta_A$ on the choice of $\delta$ is inessential, since the difference of two distance zeta functions corresponding to the same set $A$ and different values of $\d$ 
can be identified with an entire function. Note that without loss of generality (in fact, simply by replacing $A$ by its closure), we could assume that $A$ is an arbitrary (nonempty) compact subset of ${\mathbb R}^N$. Similar comments could be made about the tube zeta functions introduced in Definition \ref{zeta_tilde} below.


\subsection{Analyticity of the distance zeta functions}\label{properties_analyticity} The main result of this section is stated in Theorem~\ref{an}. It shows that the zeta function $\zeta_A$ is analytic (i.e., holomorphic) in the half-plane $\{\re s>\ov\dim_BA\}$, and that (under the mild hypotheses of part $(c)$ of Theorem \ref{an}) the lower bound is optimal.
In other words, the {\em abscissa of absolute convergence} $D(\zeta_A)$ of the Dirichlet-type integral defined by the right-hand side of (\ref{z}) is always equal to the upper box dimension of~$A$ and, under the additional hypotheses of Theorem \ref{an}(c), it also coincides with the abscissa of holomorphic continuation $D_{\rm hol}(\zeta_A)$.

In order to prove Theorem \ref{an}, we shall need a result due to Harvey\label{harvey} and Polking\label{polking} (see
\cite[p.\ 42]{acta}), 
obtained in order to study the singularities of the solutions of certain linear partial differential equations, and
which we now formulate in a different, but equivalent way: 
\begin{equation}\label{int}
\mbox{If\q $\gamma\in(-\ty,N-\ov\dim_BA)$,\q then\q $\int_{A_\delta}d(x,A)^{-\gamma}\D x<\ty$,}
\end{equation}
where $\delta$ is an arbitrary positive number. 
This result and its various extensions is discussed in \cite[Sections 3 and 4]{rae}.
For the sake of completeness, we provide an extension of (\ref{int}), which we shall need later on. We omit the proofs of the following two lemmas.
They can be obtained by using, e.g., the identity $\int_{\eR^N}f(x)^\a\D x=\a\int_0^{+\ty} t^{\a-1}|\{f>t\}|\,\D t$, where $f:\eR^N\to[0,+\ty]$ is a Lebesgue measurable function and $\a\in(0,+\ty)$ (see \cite[p.\ 198]{folland}), and by using the definition of the upper box dimension $\ov\dim_BA$ given in \eqref{dim} and \eqref{Mty} above.

\begin{lemma}\label{identity0}
Let $A$ be a bounded subset of $\eR^N$, $\d>0$ and $\gamma\in(-\ty,N-\ov\dim_BA)$. Then
\begin{equation}\label{identity}
\int_{A_\delta}d(x,A)^{-\gamma}\,\D x=\delta^{-\gamma}|A_\delta|+\gamma\int_0^\delta t^{-\gamma-1}|A_t|\,\D t.
\end{equation}
Furthermore, both of the integrals appearing in~\eqref{identity} are finite; hence, they are convergent Lebesgue integrals.
\end{lemma}

\begin{lemma}\label{optimal} Let $A$ be a bounded subset of $\eR^N$, $\delta>0 $ and $\gamma>N-\ov\dim_BA$. Then $\int_{A_\delta}d(x,A)^{-\gamma}\D x=+\ty$.
\end{lemma}

\begin{remark}\label{rm2.5}
If $\gamma:=N-\ov\dim_BA$, then the conclusion of Lemma~\ref{optimal} does not hold, in general. Indeed, a class of counterexamples is provided in \cite[Theorem 4.3]{rae}. 
\end{remark}

\medskip



\medskip

In the sequel, we shall usually say more briefly that $D(\zeta_A)$ is the {\em abscissa of convergence} of $\zeta_A$, meaning the abscissa of Lebesgue (i.e., absolute) convergence of $\zeta_A$; see \eqref{DzetaA} and the comment following it. 

\begin{theorem}\label{an} Let $A$ be an arbitrary bounded subset of $\eR^N$ and let $\delta>0$. Then$:$

\bigskip

$(a)$ The zeta function $\zeta_A$ defined by \eqref{z} is holomorphic in the 
half-plane $\{\re s>\ov\dim_BA\}$,
and for all complex numbers $s$ in that region, we have
\begin{equation}\label{zeta'}
\zeta_A'(s)=\int_{A_\delta}d(x,A)^{s-N}\log d(x,A)\,\D x.
\end{equation}

\bigskip

$(b)$  We have 
\begin{equation}
\overline\dim_BA=D(\zeta_A), 
\end{equation}
where $D(\zeta_A)$ is the abscissa of Lebesgue $($i.e., absolute$)$ convergence of $\zeta_A$. 
Furthermore, in light of part $(a)$, we always have $D_{\rm hol}(\zeta_A)\le D(\zeta_A)$.


\bigskip
$(c)$ If the box $($or Minkowski$)$ dimension $D:=\dim_BA$ exists, $D<N$, and $\M_*^D(A)>0$, then 
 $\zeta_A(s)\to+\ty$ as $s\to D^+$, $s\in\eR$. In particular, in this case, we also have that 
\begin{equation}\label{dimDDhol}
\dim_BA=D(\zeta_A)=D_{\rm hol}(\zeta_A).
\end{equation}
\end{theorem}

\begin{proof}

\bigskip 
$(a)$ Denoting the right-hand side of (\ref{zeta'}) by $I(s)$,
and choosing any $s\in\Ce$ such that $\re s>\ov\dim_BA$, it suffices to show that
\begin{eqnarray}\label{der}
 R(h)&:=&\frac{\zeta_A(s+h)-\zeta_A(s)}h-I(s)\\
 &=&\int_{A_\delta}\left(
\frac{d(x,A)^h-1}h-\log d(x,A)
\right)\,d(x,A)^{s-N}\D x\nonumber
\end{eqnarray}
converges to zero as $h\to0$ in $\Ce$, with $h\ne0$. 

Let $d:=d(x,A)\in(0,\delta)$. Defining
\begin{equation}\label{fhdef}
f(h):=\frac{d^h-1}h-\log d=\frac1h(\E^{(\log d)h}-1)-\log d,
\end{equation}
and using the MacLaurin series $\E^z=\sum_{j\ge0}\frac{z^j}{j!}$, we obtain that
\begin{equation}\label{fh}
f(h)=h(\log d)^2\sum_{k=0}^\ty \frac1{(k+2)(k+1)}\cdot\frac{(\log d)^kh^k}{k!}.
\end{equation} 
Furthermore, assuming without loss of generality that $0<\delta\le1$, and hence $\log d\le0$, we have
\begin{eqnarray}
|f(h)|&\le&\frac12|h|\,(\log d)^2\sum_{k=0}^\ty\frac{(|\log d|\,|h|)^k}{k!}\nonumber\\
&=&\frac12|h|\,(\log d)^2\E^{-(\log d)|h|}=\frac12|h|\,(\log d)^2d^{-|h|}.\nonumber
\end{eqnarray}
Therefore,
\begin{equation}
|R(h)|\le\frac12|h|\int_{A_\delta}|\log d(x,A)|^2d(x,A)^{\re s-N-|h|}\D x.
\end{equation}
Let $\e>0$ be a sufficiently small number, to be specified below. Taking $h\in\Ce$ such that $|h|<\e$, since $\delta\le1$
and hence $d(x,A)\le1$ for all $x\in A_\d$,
we have
$$
|R(h)|\le\frac12|h|\int_{A_\delta}|\log d(x,A)|^2d(x,A)^\e d(x,A)^{\re s-N-2\e}\D x.
$$
Since there exists a positive constant $C=C(\delta,\e)$ such that $|\log d|^2d^\e\le C$ for all $d\in(0,\delta)$,
we see that
\begin{equation}\label{R}
|R(h)|\le \frac12C|h|\int_{A_\delta}d(x,A)^{\re s-N-2\e}\D x.
\end{equation}
Letting $\gamma:=2\e+N-\re s$, we see that the integrability condition $\gamma<N-\ov\dim_BA$ stated in (\ref{int})
is equivalent to $\re s>\ov\dim_BA+2\e$. Observe that this latter inequality holds for all positive $\e$ small enough, due to the assumption $\re s>\ov\dim_BA$. Hence, $R(h)\to0$ as $h\to0$ in $\Ce$, with $h\ne0$. This proves part $(a)$. 
\bigskip

$(b)$ Lemma~\ref{optimal} implies that for any real number $\alpha<D=\ov\dim_BA$, we have $\int_{A_\delta}d(x,A)^{\alpha-N}\, \D x=+\infty$. On the other hand, in light of estimate (\ref{int}), we know that $\zeta_A(\alpha)=\allowbreak \int_{A_\delta}d(x,A)^{\alpha-N}\,\D x<\infty$ for any $\alpha>D$.
We therefore deduce from the definition \eqref{DzetaA} of $D(\zeta_A)$ that $D(\zeta_A)=\ov\dim_BA$. This completes the proof of part $(b)$.

\bigskip

$(c)$ Condition $\M_*^D(A)>0$ implies that for any fixed $\delta>0$ 
there exists $C>0$
such that for all $t\in(0,\delta)$, we have $|A_t|\ge Ct^{N-D}$. Using (\ref{int}) and Lemma \ref{identity0}, we see that for any
$\gamma\in(0,N-D)$, 
\begin{eqnarray}
\ty&>&I(\gamma):=\int_{A_\delta}d(x,A)^{-\gamma}\D x=\delta^{-\gamma}|A_\delta|+\gamma\int_0^\delta t^{-\gamma-1}|A_t|\,dt\nonumber\\
&\ge&\gamma C\int_0^\delta t^{N-D-\gamma-1}\D t= \gamma C\frac{\delta^{N-D-\gamma}}{N-D-\gamma}.\nonumber
\end{eqnarray}
Therefore, if $\gamma\to N-D$ from the left, then $I(\gamma)\to+\ty$. Equivalently, if $s\in\eR$
is such that $s\to D^+$, then $\zeta_A(s)\to+\ty$. Hence, $\zeta_A$ has a singularity at $s=D$. Since, in light of part $(a)$,
we know that $\zeta_A$ is holomorphic for $\re s>D$, we deduce that $\{\re s>D\}$ is the maximal right half-plane to which $\zeta_A$
can be holomorphically continued; i.e., $\mathcal{H}(\zeta_A)=\{\re s>D\}$ and so $D_{\rm hol}(\zeta_A)=D$. Since, in light of part $(b)$ (and because $\dim_BA$ exists, according to the assumptions of part $(c)$),
$D:=\dim_BA=D(\zeta_A)$, we conclude that \eqref{dimDDhol} holds and hence, the proof of part $(c)$ is complete. This concludes the proof of the theorem.
\end{proof}

\begin{remark}
An alternative proof of part $(a)$ of Theorem \ref{an} can be given by using a well-known theorem concerning the holomorphicity of functions defined by integrals on $A_\d$ depending holomorphically on a parameter. In applying this theorem (see [LapRa\v Zu1] and the text of Definition \ref{abscissa_f} below) one needs to use the (obvious) fact according to which the function $x\mapsto d(x,A)$ is bounded from above (by $\d$); in other words, $\zeta_A$ (as defined by \eqref{z}) is a tamed DTI (in the sense of Definition \ref{abscissa_f} below).
\end{remark}

Next, we comment on some of the hypotheses and conclusions of Theorem~\ref{an}.

\begin{remark}
$(i)$ The condition $\M_*^D(A)>0$ in the hypotheses of Theorem
~\ref{an}(c) cannot be omitted.
Indeed, for $N=1$, there exists a class of subsets $A\st[0,1]$ such that $D=\dim_BA$ exists and $\M_*^D(A)=0$, while $\zeta_A(D)=\int_{A_\delta}d(x,A)^{D-N}\D x<\ty$; see \cite[Theorem 4.3]{rae}.

This class of bounded subsets of $\eR$ can be easily extended to $\eR^N$ for any $N\ge2$ by letting $B:=A\times[0,1]^{N-1}\st[0,1]^N$ and using the results of Subsection~\ref{dtx}.

\smallskip

$(ii)$ The inequality $D_{\rm hol}(\zeta_A)\le D(\zeta_A)$ is sharp. Indeed, there exist compact subsets of $\eR^N$ such that $D_{\rm hol}(\zeta_A)=D(\zeta_A)$.
For example, $A=C\times[0,1]^{N-1}$, where $C$ is the ternary Cantor set or, more generally, $C=\pa\O$ is the boundary of any (nontrivial) bounded fractal string $\O\st\eR$. 
(In that case, we have $D_{\rm hol}(\zeta_A)=D(\zeta_A)=\ov\dim_BA=N-1+\dim_B C$.)
This follows from Theorem \ref{1.2.31} in Subsection \ref{zeta_s} below and the comment following it.

\smallskip

$(iii)$ The assumptions of part $(c)$ of Theorem \ref{an} are satisfied by most fractals of interest to us. (One notable exception is the boundary $A$ of the Mandelbrot set (viewed as as a subset of $\eR^2\simeq\Ce$), for which $\dim_HA=2$ (and consequently, $\dim_BA=2$
 since $\dim_HA\le\dim_BA$), according to Shishikura's well-known theorem \cite{shishikura}.) We note that, on the other hand, there exists a bounded subset of $\eR^N$ not satisfying the hypotheses of part $(c)$ of Theorem \ref{an} and such that $D_{\rm hol}(\zeta_A)<D(\zeta_A)$. Indeed, an easy computation shows that, for example, for $N=1$ and $A=[0,1]$, we have that $D_{\rm hol}(\zeta_A)=0$ and $D(\zeta_A)=1$.
At present, however, we do not know whether there exist nontrivial subsets $A$ of $\eR$ (or, more generally, of $\eR^N$) for which $D_{\rm hol}(\zeta_A)<D(\zeta_A)$.
\end{remark}

\subsection{Zeta functions of fractal strings and of associated fractal sets}\label{zeta_s}\label{properties_zeta}
In Example \ref{L} below, we show that Definition \ref{defn} provides a natural extension of the zeta function associated with a (bounded) {\em fractal string} $\mathcal L=(\ell_j)_{j\ge1}$, where $(\ell_j)_{j\geq 1}$ is a nonincreasing sequence of positive numbers such that $\sum_{j=1}^\infty \ell_j<\ty$:
\begin{equation}\label{string}
\zeta_{\mathcal L}(s)=\sum_{j=1}^\infty \ell_j^s,
\end{equation}
for all $s\in\Ce$ with $\re s$ sufficiently large. Note that the sequence $(\ell_j)_{j\ge1}$ of positive numbers is assumed to be infinite.

The study of zeta functions of fractal strings arose naturally in the early 1990s in joint work of the first author\label{lapidus2} with Carl Pomerance 
[{LapPo1--2}] 
and with Helmut Maier\label{maier} \cite{LapMa2} 
 while investigating direct and inverse spectral problems associated with the vibrations of a fractal string. Such a zeta function, $\zeta_{\mathcal L}$, called the {\em geometric zeta function}\label{geometric_zf} of $\mathcal L$, has since then been studied in a number of references, including the monograph \cite{lapidusfrank12}..
(See the broader list of references given in the introduction.)

Recall that, geometrically, a (bounded) {\em fractal string}\label{fr_str} is a bounded open set $\Omega\subseteq\eR$. It can be uniquely written as a disjoint union of open intervals $I_j$ 
($\Omega=\cup_{j=1}^\infty I_j$) with lengths $\ell_j$ (i.e., $\ell_j=|I_j|$ for all $j\ge1$). Without loss of generality, one may assume that $(\ell_j)_{j\geq 1}$ is 
written in nonincreasing order and that $\ell_j\to 0$ as $j\to\infty$: $\ell_1\geq \ell_2\geq\cdots$. In order to avoid trivial special cases, we will assume implicitly throughout this paper that $\mathcal L$ is nontrivial; i.e., that $\mathcal L$ consists of an infinite sequence of lengths (or `scales') and hence, that $\O$ does not consist of a finite union of bounded open intervals. If $\mathcal L$ is trivial, then we must replace $D_{\rm hol}(\zeta_{\mathcal L})$ by $\max\{D_{\rm hol}(\zeta_{\mathcal{L}}),0\}$ in \eqref{2.201/2} of Theorem \ref{2.181/2} (since then, $D_{\rm hol}(\zeta_{\mathcal L})=-\ty$ and $D(\zeta_{\mathcal L})=\d_{\pa\O}\ge0$). 
From the point of view of fractal string theory, 
one may identify a fractal string with the sequence $\mathcal L$ of its lengths $($or {\em scales}$)$: $\mathcal L=(\ell_j)_{j\geq 1}$. The bounded open set $\O$ is then called a {\em geometric realization of $\mathcal{L}$}. Note that $|\O|=\sum_{j=1}^\ty \ell_j<\ty$, where $|\O|=|\O|_1$ denotes the $1$-dimensional Lebesgue measure (or length) of $\O$.

We now recall a basic property of $\zeta_{\mathcal L}$, first observed in \cite{Lap2}, using a key result of Besicovich\label{besicovich} and Taylor\label{taylor} \cite{BesTay}. (For a direct proof, see \cite[Theorem 1.10]{lapidusfrank12}.)

\begin{theorem}\label{2.181/2}
If $\mathcal L$ is a nontrivial bounded fractal string $($i.e., $\mathcal L=(\ell_j)_{j\geq 1}$ is an infinite sequence$)$, then the abscissa of convergence $D(\zeta_{\mathcal L})$ of $\zeta_{\mathcal L}$ coincides with the $($inner$)$ Minkowski dimension $\delta_{\partial\Omega}$ of $\partial\mathcal L=\partial\Omega:$
\begin{equation}\label{2.201/2}
D(\zeta_{\mathcal L})=D_{\rm hol}(\zeta_{\mathcal L})=\delta_{\partial\Omega}.
\end{equation}
\end{theorem}

Recall that, by definition,
\begin{equation}\label{Ddef}
D(\zeta_{\mathcal L}):=\inf\Big\{\alpha\in\eR\ :\ \sum_{j=1}^{\infty}\ell_j^{\alpha}<\infty\Big\},
\end{equation}
while $\delta_{\partial\Omega}$ is then defined in terms of the volume (i.e., length) of the inner epsilon (or tubular) neighborhoods of $\pa\Omega$, namely, $(\pa\Omega)_{\e}\cap\Omega=\{x\in\Omega\ :\ d(x,\partial\Omega)<\e\}$; see \cite[Chapter 1]{lapidusfrank12}.

In order to establish the equality $D(\zeta_{\mathcal L})=D_{\rm hol}(\zeta_{\mathcal L})$ from Theorem \ref{2.181/2}, one first notes that 
$\zeta_{\mathcal L}$ is holomorphic for $\re s>D(\zeta_{\mathcal L})$ and that $\{\re s>D(\zeta_{\mathcal L})\}$ is the largest open right half-plane having this property; i.e., $D(\zeta_{\mathcal L})=D_{\rm hol}(\zeta_{\mathcal{L}})$. The latter property follows from the fact that (because $\zeta_{\mathcal L}(s)$ is initially given in \eqref{string} by a Dirichlet series with positive coefficients), $\zeta_{\mathcal L}(s)\to+\ty$ as $s\to D^+$, $s\in\eR$, where $D:=D(\zeta_{\mathcal L})=\d_{\pa\O}$; see, e.g., \cite[Section VI.2.3]{serre}.
The proof of the equality $D(\zeta_{\mathcal L})=\d_{\pa\O}$ requires significantly more work; see the aforementioned references.

Note that, more precisely, $\ov\dim_BA_{\mathcal L}=\d_{\pa\O}$ is equal to $\ov\dim_B(\pa\O,\O)$, the Minkowski dimension of $\pa\O$ relative to $\O$
(also called the inner Minkowski dimension of $\pa\O$, or, equivalently, of $\mathcal L$) which is defined in terms of the volume (i.e., length) of the inner tubular neighborhoods of $\O$. More specifically, $\d_{\pa\O}$ is given by \eqref{dim} or \eqref{Mty}, except for $|A_t|$ replaced by $|A_t\cap\O|_1$, with $A:=\pa\O$, in the counterpart of the second equality of \eqref{mink}.

\medskip


In fractal string theory, one is particularly interested in the meromorphic continuation of $\zeta_{\mathcal  L}$ to a suitable region (when it exists), along with its poles, which are called the {\em complex dimensions}\label{c_dim} of $\mathcal L$.
In particular, in \cite{lapidusfrank12}, explicit formulas are obtained that are  applicable to various counting functions associated with the geometry and the spectra of fractal strings,
as well as to $|(\pa\O)_t\cap\O|_1$, now defined as the volume of the inner tubular neighborhood of $\pa\O$ (i.e., of $\mathcal L$).
These explicit formulas are expressed in terms of the complex dimensions (i.e., the poles of $\zeta_{\mathcal L}$) and the associated residues. Furthermore, they enable one to obtain a very precise understanding of the oscillations underlying the geometry and spectra of fractal strings (as well as of more general fractal-like objects).

From the perspective of the theory developed in the present work, a convenient choice for the set $A_{\mathcal L}$ corresponding to the fractal string $\mathcal L=(\ell_j)_{j\ge1}$ is  
\begin{equation}\label{A_L}
A_{\mathcal L}:=\{a_k : k\geq 1\},\q \mbox{where\q$a_k:=\sum_{j\ge k}\ell_j$\q for each\q $k\geq 1$.}
\end{equation}
As follows easily from Theorem \ref{2.181/2} and the definition of $A_{\mathcal L}$ (see Equations \eqref{cantor_string} below) and \eqref{simeq}, the function $\zeta_{\mathcal L}$ in (\ref{string}) is holomorphic for all $s\in\Ce$ with $\re s>\ov\dim_BA_{\mathcal L}$.
Moreover, this bound 
is optimal. In other words, $\overline{\dim}_BA$ coincides both with the abscissa of holomorphic continuation $D_{\rm hol}(\zeta_{\mathcal L})$ and the abscissa of (absolute) convergence $D(\zeta_{\mathcal L})$ of $\mathcal L$.
Furthermore, $\zeta_{\mathcal L}(s)\to+\ty$ as $s\in\eR$ converges to $\ov\dim_BA_{\mathcal L}$ from the right; compare with Theorem \ref{an} above.
In light of Theorem $\ref{2.181/2}$, Theorem \ref{an}(b), Equation $(\ref{2.201/2})$ and Equations \eqref{cantor_string}--\eqref{simeq},
we then have the following equalities:
\begin{equation}\label{eq}
\ov\dim_BA_{\mathcal L}=D(\zeta_{A_{\mathcal L}})=D_{\rm hol}(\zeta_{A_{\mathcal L}})=D(\zeta_{\mathcal L})=D_{\rm hol}(\zeta_{\mathcal L})=\d_{\pa\O}.
\end{equation}

The following example shows that the study of the geometric zeta function $\zeta_{\mathcal L}$ of any (bounded) fractal string $\mathcal L$ can be reduced to the study of the distance zeta
 function $\zeta_{A_{\mathcal L}}$ of the associated bounded set $A_{\mathcal L}$ on the real line. (See also Remark \ref{entirely} below for a more general statement.)

\medskip

\begin{example}\label{L} Let $(I_k)_{k\ge1}$ be a sequence of bounded intervals, $I_k=(a_{k+1},a_k)$, $k\ge1$, where the $a_k$'s are defined by \eqref{A_L}, and let $s$ be a complex variable.
Using (\ref{z}), we see that the distance zeta function of $A=A_{\mathcal L}$ for $\re s>D(\zeta_{\mathcal L})$ is given by 
\begin{equation}\label{zetal}
\zeta_A(s)=2\int_0^\delta x^{s-1}\D x+\sum_{k=1}^\infty\int_{I_k}d(x,\pa I_k)^{s-1}\D x\\
=2s^{-1}\delta^s+\sum_{k=1}^\ty J_k(s),
\end{equation}
where the first term in this last expression corresponds to the boundary points of the interval $(0,a_1)$.
Assuming that $\delta\ge \ell_1/2$, we have that for all $k\ge1$,
\begin{equation}\label{zetalk}
J_k(s)=s^{-1}2^{1-s}\ell_k^s.
\end{equation}
Note that we assume that $s\in\Ce$ is such that $\re s>D(\zeta_{\mathcal L})$, so that the series $\sum_{k=1}^\ty J_k(s)$ appearing in \eqref{zetal} is convergent.
 In light of (\ref{string})--(\ref{A_L}) and \eqref{zetal}, we then obtain the following relation:
\begin{equation}\label{cantor_string}
\zeta_A(s)=s^{-1}2^{1-s}\zeta_{\mathcal L}(s)+2s^{-1}\delta^{s}.
\end{equation}
The case when $0<\delta<\ell_1/2$ yields an analogous relation:
\begin{equation}\label{simeq}
\zeta_A(s)=u(s)\zeta_{\mathcal L}(s)+v(s),
\end{equation}
where again $u(s):=s^{-1}2^{1-s}$, with a simple pole at $s=0$. 
Note that here, $u(s)$ and $v(s)=v(s,\delta)$ are holomorphic functions in the right half-plane $\{\re s>0\}$. 
Hence, by the principle of analytic continuation and
since $\zeta_{\mathcal L}$ is holomorphic for $\re s>\ov\dim_BA$, the same relation still holds for the meromorphic extensions of $\zeta_A$ and of $\zeta_{\mathcal L}$ (when they exist, see Theorem \ref{1.2.31}) within the right half-plane $\{\re s>0\}$. 
\end{example}

The following result is in accordance with Theorem~\ref{2.181/2}.

\begin{theorem}\label{1.2.31}
Let $\mathcal{L}=(\ell_j)_{j\ge1}$ be a $($nontrivial$)$ fractal string such that $\sum_{j\ge1}\ell_j<\ty$, and 
let $A_{\mathcal L}=\big\{a_k=\sum_{j\ge k}\ell_j:k\ge1\big\}$. Then 
\begin{equation}\label{DAL}
D(\zeta_{A_{\mathcal L}})=D_{\rm hol}(\zeta_{A_{\mathcal L}})=D(\zeta_{\mathcal L})=D_{\rm hol}(\zeta_{\mathcal L})=\ov\dim_BA_{\mathcal L}.
\end{equation}
Furthermore, given $c\ge0$, the sets of poles of the meromorphic extensions of $\zeta_{A_{\mathcal L}}$ and $\zeta_{\mathcal L}$
$($if one, and therefore both, of the extensions exist$)$
to the open right half-plane $\{\re s>c\}$ coincide. Moreover, the poles of $\zeta_{A_{\mathcal{L}}}$ and $\zeta_{\mathcal{L}}$ $($in such a half-plane$)$ have the same multiplicities. 

More generally, given any subdomain $U$ of $\Ce\setminus\{0\}$ containing the critical line $\{\re s=D(\zeta_{\mathcal L})\}$, $\zeta_{A_{\mathcal L}}$ has a meromorphic continuation to $U$ if and only $\zeta_{\mathcal L}$ does, and in that case, $\zeta_{A_{\mathcal L}}$ and $\zeta_{\mathcal L}$ have the same visible poles in $U$ and with the same multiplicities.
\end{theorem}

\begin{proof}
The first claim follows from Theorem~\ref{2.181/2} combined with parts $(a)$ and $(b)$ of Theorem~\ref{an}. The second and the third claims are an immediate consequence of the identities~\eqref{cantor_string} and (\ref{simeq}) in Example \ref{L}.
\end{proof}

\begin{remark}\label{entirely} An entirely similar proof shows that, in Example \ref{L} and Theorem \ref{1.2.31}, we can replace $A_{\mathcal L}$ with $A:=\pa\O$, where the bounded open set $\O\st\eR$
is any geometric realization of the (nontrivial) fractal string $\mathcal L$, provided $\ov\dim_BA:=\d_{\pa\O}$, as defined in the comment following \eqref{Ddef}.
Hence, with the notation used in \eqref{Ddef}, we also have the following counterpart of \eqref{eq} in this more general situation:
\begin{equation}
D(\zeta_{\mathcal L})=D_{\rm hol}(\zeta_{\mathcal L})=D(\zeta_{\pa\O})=D_{\rm hol}(\zeta_{\pa\O})=\d_{\pa\O}:=\ov\dim_B(\pa\O,\O).
\end{equation}
\end{remark}

Actually, a direct computation shows that, in that case, the relation between $\zeta_{\mathcal L}$ and $\zeta_{\pa\O,\O}$ (the distance zeta function of the fractal string $\mathcal L$, viewed as a relative fractal drum, in the sense of \cite{memoir}, is even more straightforward:
\begin{equation}
\zeta_{\pa\O,\O}(s)=\frac{2^{1-s}}{s}\zeta_{\mathcal L}(s),
\end{equation}
for every $s\in\Ce$ such that $\re s>\d_{\pa\O}$ and, more generally, in every domain of $\Ce$ to which one (and hence both) of these two fractal zeta functions can be meromorphically continued.

\subsection{Equivalent zeta functions}\label{eqzf}
In this subsection, we shall introduce an equivalence relation $\sim$ on the set of zeta functions
(see Definition~\ref{equ}). Let us illustrate its purpose in the case 
of the distance zeta function $\zeta_A$ of a given nonincreasing infinite sequence $A=(a_k)_{k\ge1}$, converging to zero in $\eR$.
As we saw in Example \ref{L}, it makes sense to identify it with its simpler form $\zeta_{\mathcal L}$, where $\mathcal L=(\ell_j)_{j\ge1}$
is the associated bounded fractal string, defined by 
$\ell_j=a_j-a_{j+1}$. This is done by removing 
the inessential functions $u(s)$ and $v(s)$ appearing in Equation~(\ref{simeq}) above. Therefore, $\zeta_A\sim\zeta_{\mathcal L}$.

Throughout this subsection (and Appendix A in which this topic is further developed), we will assume that $E$ is a locally compact, Hausdorff topological (and metrizable) space and that $\mu$ is a {\em local} (roughly speaking, locally bounded) positive or complex measure (in the sense of \cite{dolfr}, \cite{johlap}, or \cite[Chapter 4]{lapidusfrank12}). In short, a {\em local measure} is a $[0,+\ty]$-valued or $\Ce$-valued set-function on $\mathcal{B}:=\mathcal{B}(E)$ (the Borel $\s$-algebra of $E$), whose restriction to $\mathcal{B}(K)$, 
where $K$ is an arbitrary compact subset of $E$,
is a bounded positive measure or is a complex (and hence, bounded) measure, respectively.
The {\em total variation measure} of $\mu$ (see, e.g., [Coh] or [Ru]) is denoted by $|\mu|$; it is a (local) positive measure and, if $\mu$ is itself positive, 
then $|\mu|=\mu$. We refer to [Coh, Fol, Ru] for the theory of standard positive or complex measures.

We assume that the $\mu$-measurable function $\f:E\to\eR\cup\{+\ty\}$ appearing in Definition \ref{abscissa_f} just below is {\em tamed}, in the following sense: there exists a positive constant $C=C(\f)$ such that
\begin{equation}\label{E1}
|\mu|(\{\f>C\})=0;
\end{equation}
i.e.,  $\f$ is essentially bounded from above with respect to $|\mu|$. We then say that $f$, defined by \eqref{Efi} below, is a {\em tamed} DTI.

\begin{defn}\label{abscissa_f}
Given a {\em tamed Dirichlet-type integral} (tamed DTI, in short) function $f=f(s)$ of the form 
\begin{equation}\label{Efi}
f(s):=\int_E\f(x)^s\,\D\mu(x),
\end{equation}
where $\mu$ is a suitable (positive or complex) local (i.e., locally bounded) measure on a given (measurable) space $E$ [i.e., $\mu:\mathcal{B}\to[0,+\ty]$ or $\mu:\mathcal{B}\to\Ce$],
and $\f:E\to\eR\cup\{+\ty\}$ is a $\mu$-measurable function such that $\f\ge0$ $\mu$-a.e.\ on $E$,
we define the {\em abscissa of $($absolute$)$ convergence} $D(f)\in\eR\cup\{\pm\ty\}$ 
by
\begin{equation}\label{D(f)}
\begin{aligned}
 D(f)&:=\inf\left\{\alpha\in\eR:\int_E\f(x)^\a\D|\mu|(x)<\ty\right\}\\
&\phantom{:}=\inf\big\{\alpha\in\eR:\mbox{$\f(x)^s$ is Lebesgue integrable for $\re s>\a$}\big\}.
\end{aligned}
\end{equation}
It follows that the {\em half-plane of $($absolute$)$ convergence of $f$}, namely, $\Pi(f):=\{\re s>D(f)\}$, is the {\em maximal} open
right half-plane (of the form $\{\re s>\a\}$, for some $\a\in\eR\cup\{\pm\ty\}$) on which the function $x\mapsto\f(x)^s$ is absolutely (i.e., Lebesgue)
integrable. (Note that $D(f)$ is well defined for any tamed Dirichlet-type integral $f$.)

In \eqref{D(f)}, by definition, $\inf\emptyset:=+\infty$ and $\inf\eR=-\ty$. Using a classic theorem about the holomorphicity of integrals depending analytically on a parameter, one can show that $f$ is holomorphic on $\{\re s>D(f)\}$. Hence, it follows that $D_{\rm hol}(f)\le D(f)$. Here, $D_{\rm hol}(f)\in\eR\cup\{\pm\ty\}$, the {\em abscissa of holomorphic continuation of $f$}, is defined exactly as $D_{\rm hol}(\zeta_A)$ in \eqref{Dhol}, except for $\zeta_A$ replaced by $f$. 

In \eqref{D(f)}, the integral is taken with respect to $|\mu|$, the total variation measure of $\mu$; recall that if $\mu$
is positive, then $|\mu|=\mu$. Note that we may clearly replace $\f(x)^s$ by $\f(x)^{\re s}$ in the second equality of~\eqref{D(f)},
since for a measurable function, Lebesgue integrability is equivalent to absolute integrability.
\end{defn}

\begin{remark}
There are many examples for which $D_{\rm hol}(f)=D(f)$ (see, e.g., Equation \eqref{DAL} in Theorem \ref{1.2.31} or Equation~\eqref{dimDDhol} in Theorem~\ref{an}) and other examples for which $D_{\rm hol}(f)<D(f)$ (this is so for Dirichlet $L$-functions with a nontrivial primitive character, in which case $D_{\rm hol}(f)=-\ty$ but $D(f)=1$; see, e.g., \cite[Section VI.3]{serre}).
This is the case, for instance, if $f(s):=\sum_{n=1}^\ty(-1)^{n-1}/n^s$.
\end{remark}

\begin{remark}
All of the fractal zeta functions encountered in this work, namely, the distance and tube zeta functions (see Subsection \ref{properties_definition}  above and Subsection \ref{residues_m_tube} below), their counterparts for relative fractal drums, the geometric zeta function of (possibly generalized) fractal strings
 (\cite[Chapters 1 and 4]{lapidusfrank12}), as well as the spectral zeta functions of (relative) fractal drums (see \cite{Lap3,brezish}) are tamed DTIs; i.e., they are Dirichlet-type integrals (in the sense of \eqref{Efi}, and for a suitable choice of set $E$, function $\f$ and measure $\mu$) satisfy condition \eqref{E1}. This justifies, in particular, the use of the expression ``abscissa of (absolute) convergence'' and ``half-plane of (absolute) convergence'' for all of these fractal zeta functions, including the tube and distance zeta functions which are key objects in the present paper. 

For example, for the distance zeta function $\zeta_A$ (as in Definition \ref{defn} above), we can choose $E:=A_\d$ (or else, $E:=A_\d\setminus\ov A$), $\f(x):=d(x,A)$ for $x\in E$ and $\mu(\D x):=d(x,A)^{-N}\D x$, while for the tube zeta function (as in Definition \ref{zeta_tilde} below), we  can choose $E:=(0,\d)$, $\f(t):=t$ for $t\in E$ and $\mu(\D x):=t^{-N-1}|A_t|\D t=t^{-N}|A_t|\,(\D t/t)$. In both cases, it is easy to check that the tameness condition \eqref{E1} is satisfied, with $C:=\d$.

In closing, we note that the class of tamed Dirichlet-type integrals also contains all arithmetic zeta functions (that is, all zeta functions occurring in number theory); see, e.g., [ParSh1--2, Pos, Ser, Tit, Lap-vFr2, Lap4].
\end{remark}



Recall from part (b) of Theorem \ref{an} that we have the following result, which is very useful for the computation of the upper box dimension of fractal sets.

\begin{cor}\label{an1}
Let $A$ be any bounded subset of $\eR^N$. Then
\begin{equation}
\ov\dim_BA=D(\zeta_A).
\end{equation}
Hence, we have $0\le D(\zeta_A)\le N$.
\end{cor}

Following \cite[Sections 1.2.1 and 5.1]{lapidusfrank12}, assume that the set $A$ has the property that $\zeta_A$ can be extended to a meromorphic function defined on $G\stq\Ce$, 
where $G$ is an open and connected neighborhood of the {\em window}\label{window} $\bm W$ defined by
$$
\bm{W}:=\{s\in\Ce: \re s\ge S(\im s)\}.
$$
Here, the function $S:\eR\to(-\ty,D(\zeta_A)]$, called the {\em screen},\label{screen} is assumed to be Lipschitz continuous.
Note that the closed set $\bm W$ contains the {\em
critical line} (of convergence) $\{\re s=D(\zeta_A)\}$.\label{cr_line_w} In other words, we assume that $A$ is such that its distance zeta function can be extended meromorphically  to an open domain $G$ 
containing the closed right half-plane $\{\re s\ge D(\zeta_A)\}$. (Following the usual conventions, we still denote by $\zeta_A$ the meromorphic continuation of $\zeta_A$ to $G$, which is necessarily unique due to the principle of analytic continuation. Furthermore, as in \cite{lapidusfrank12}, we assume that the {\em screen} 
\begin{equation}
\bm{S}:=\pa\bm{W}=\{S(\tau)+\I\tau:\tau\in\eR\}
\end{equation} 
does not contain any poles of $\zeta_A$.) 
A set $A$ satisfying this property 
is said to be {\em admissible}\label{admissible}. (There exist nonadmissible fractal sets; see Subsection \ref{hyperfractal}.)
The notion of admissibility used here is weaker than the one used in \cite{cras2} and \cite{mm} because we do not establish fractal tube formulas in this paper.

We will also need to consider {\em the set of poles of $\zeta_A$ located on the critical line} $\{\re s=D(\zeta_A)\}$, where $D(\zeta_A)$ is assumed to be a real number (see Definition~\ref{dimc}):
\begin{equation}\label{po}
\po_c(\zeta_A)=\{\omega\in \bm{W}:\mbox{$\omega$ is a pole of $\zeta_A$ and $\re \omega=D(\zeta_A)$}\}.
\end{equation}
It is a subset of {\em the set of all poles of} $\zeta_A$ in $\bm W$, that we denote by $\po(\zeta_A)$ or $\po(\zeta_A,\bm{W})$ (see Definition \ref{1.331/2}). 

\begin{remark}\label{-ty}  We assume in the definition of $\po_c(\zeta_A)$ that $D(\zeta_A)\in\eR$, which is the case for example if $A$ is bounded, according to Corollary~\ref{an1}. Note that clearly (and in contrast to $\po(\zeta_A)=\po(\zeta_A,\bm{W})$, to be introduced in Definition~\ref{1.331/2}), $\po_c(\zeta_A)$ is independent of the choice of the window $\bm W$.
\end{remark}

The following definition
is a slight modification of the notion of complex dimension for fractal strings.

\begin{defn}\label{dimc}
Let $A$ be an admissible subset of $\eR^N$ such that $D(\zeta_A)\in\eR$.
Then, the {\em set of principal complex dimensions} of $A$, denoted by $\dim_{PC} A$,
is defined as the set of poles of $\zeta_A$ which are located on the critical line $\{\re s=D(\zeta_A)\}$:
\begin{equation}
\dim_{PC} A:=\po_c(\zeta_A),
\end{equation}
where $\po_c(\zeta_A)$ is given by (\ref{po}).
\end{defn}

As we see, in Definition \ref{dimc}, if $A\subset\eR^N$ is bounded, the singularities of $\zeta_A$ we are interested in are located on the vertical line $\{\re s=\ov\dim_BA\}$. 


\begin{defn}\label{1.331/2}
Let $A$ be an admissible subset of $\eR^N$. Then, the {\em set of visible complex dimensions} of $A$ {\em with respect to a given window $\bm W$} (often called, in short, the {\em set of complex dimensions of $A$ relative to $\bm W$}, or simply the {\em set of $($visible$)$ complex dimensions} of $A$ if no ambiguity may arise or if $\bm{W}=\Ce$), is defined as the set of all the poles of $\zeta_A$ which are located in the window $\bm W$:
\begin{equation}\label{1.401/2}
\mathcal{P}(\zeta_A)=\{\omega\in \bm{W} : \omega\textrm{ is a pole of } \zeta_A\}.
\end{equation}
Instead of $\mathcal{P}(\zeta_A)$, we can also write $\mathcal{P}(\zeta_A,\bm{W})$, in order to stress that this set depends on $\bm W$ as well. Furthermore, all the sets of complex dimensions appearing in this paper are interpreted as multisets, i.e., with the multiplicities of the poles taken into account
\end{defn}

Next, we would like to extend the class of zeta functions to which a slight modification of Definition \ref{dimc} and Definition \ref{1.331/2} can be applied. 
Given a meromorphic function $f$ on a domain $G\subseteq\Ce$ containing the vertical line $\{\re s=D(f)\}$ (as in Remark \ref{-ty} above, we assume here that $D(f)\in\eR$), 
and which (for all $s\in \Ce$ with $\re s$ sufficiently large) is given by a convergent Dirichlet-type integral of the form \eqref{Efi} and satisfying condition \eqref{E1},
so that $D(f)<\ty$ is well defined by \eqref{D(f)}),
we define the set 
$\po_c(f)$ in much the same way as in~(\ref{po}):
\begin{equation}\label{pof}
\po_c(f)=\{\omega\in G:\mbox{$\omega$ is a pole of $f$ and $\re \omega=D(f)$}\}.
\end{equation}
It is a subset of the set $\po(f)$ of all the poles of $f$ belonging to $G$. In other words,
\begin{equation}\label{1.411/2}
\mathcal{P}(f)=\{\omega\in G:\omega\textrm{ is a pole of } f\}.
\end{equation}

\begin{remark}\label{1.333/4}
If $f=\zeta_A$, where $A$ is an admissible set for a given window $\bm W$, then (with $G:=\mathring{\bm{W}}$, the interior of the window) $\po_c(f)=\po_c(\zeta_A)$, the set of principal complex dimensions of $A$, while $\po(f,\mathring{\bm{W}})=\po(f)=\po(\zeta_A)=\po(\zeta_A,\bm{W})$, the set of (visible) complex dimensions of $A$ (relative to $\bm W$). This follows from the fact that since $A$ is admissible, $\zeta_A$ does not have any poles along the screen $\bm S$.
\end{remark}

\begin{remark}\label{1.334/5}
Observe that $\po_c(f)$ is independent of the choice of the domain $G$ containing the vertical line $\{\re s=D(f)\}$. Moreover, since as was noted earlier, the function $f$ is holomorphic for $\re s>D(f)$, there are no poles of $f$ located in the open half-plane $\{\re s>D(f)\}$; this is why we could equivalently require that the domain $G\subseteq\Ce$ contains the closed half-plane $\{\re s\geq D(f)\}$ in order to define $\po_c(f)$ and $\po(f)$. 

Finally, we note that since $\po(f)$ is the set of poles of a meromorphic function, it is a discrete subset of $\Ce$; in particular, it is at most countable. Since $\po_c(f)\stq\po(f)$, the same is true for $\po_c(f)$.
(An entirely analogous comment can be made about $\po_c(\zeta_A)$ and $\po(\zeta_A)$ in Definition \ref{dimc} and Definition \ref{1.331/2}, respectively.)
\end{remark}

We next define the equivalence of a given distance zeta function $f$ to a suitable meromorphic function $g$ (of a preferably simpler form), a notion which will be useful to us in the sequel.  Note that the relation $\sim$ introduced in Definition \ref{equ} is clearly an equivalence relation on the set of all tamed DTIs.

\begin{defn}\label{equ}
Let $f$ and $g$ be tamed Dirichlet-type integrals, as in Definition \ref{abscissa_f}, both admitting a (necessarily unique) meromorphic extension to an open connected subset $U$ of $\Ce$ which contains
the closed right half-plane $\{\re s\ge D(f)\}$. (As follows from the complete definition, this closed half-plane is actually the closure of the common half-plane of convergence of $f$ and $g$, given by $\Pi:=\Pi(f)=\Pi(g)$.) Then, the function $f$ is said to be {\em equivalent} to $g$, and we write $f\sim g$,
if $D(f)=D(g)$ (and this common value is a real number) and furthermore, the sets of poles of $f$ and $g$, located on the common critical line $\{\re s=D(f)\}$, coincide. Here, the multiplicities of the poles should be taken into account. In other words, we view the set of principal poles ${\mathcal{P}}_c(f)$ of $f$ as a multiset.
More succinctly,
\begin{equation}\label{equ2}
f\sim g\quad\overset{\mbox{\tiny def.}}\Longleftrightarrow\quad D(f)=D(g)\,\,(\in\eR)\q \mathrm{and}\q \po_c(f)=\po_c(g).
\end{equation}
\end{defn}

\medskip

If a tamed Dirichlet-type integral $f$ is given (for example, a distance zeta function $\zeta_A$ corresponding to a given fractal set $A$), the aim is to find an equivalent meromorphic function $g$, defined by a simpler expression.
Satisfactory results can already be obtained with functions $g$ of the form $g(s)=u(s)f(s)+v(s)$, for a suitable choice of the holomorphic functions $u$ and~$v$, with $u$ nowhere vanishing in the given domain, as we have seen in Example \ref{L}.

We refer to Definition A.2 in Appendix A to this paper for an extension of Definition \ref{equ} to the broader class of extended Dirichlet-type integrals (extended DTIs, for short), as introduced in Definition A.1. 

We also refer to Definition A.6 (and the comments surrounding it) at the end of Appendix A for a closely related, but somewhat different (and perhaps more practical) definition, allowing the meromorphic function $g$ not to be a DTI (or more generally, an EDTI of type I, in the terminology of Appendix A). These new definitions (Definitions A.2 and A.6) can be applied to (essentially) all the examples of interest in this paper and in our general theory. Towards the end of Appendix A, the interested reader can find a large class of functions $g$ giving the ``leading behavior'' of fractal zeta functions $f$. 
(See Theorem A.3 in Appendix A, along with its consequences.)

In the following proposition, we consider the dependence of the distance zeta function $\zeta_A$ on $\d>0$.
For this reason, we denote $\zeta_A$ by $\zeta_A(\,\cdot\,,A_{\d})$.

\begin{prop}\label{O}
Let $A$ be a bounded subset of $\eR^N$.
Then, for any two positive real numbers $\d_1$ and $\d_2$, we have
$\zeta_A(\,\cdot\,,A_{\d_1})\sim\zeta_A(\,\cdot\,,A_{\d_2})$, in the sense of Definition \ref{equ}.
\end{prop}

\begin{proof} We assume without loss of generality that $\d_1<\d_2$, since for $\d_1=\d_2$ there is nothing to prove. 
For $\re s>\ov\dim_BA$,
the difference of the functions $\zeta_A(s,A_{\d_2})$ and $\zeta_A(s,A_{\d_1})$
 is equal to
\begin{equation}\label{int12}
\int_{A_{\d_2}\setminus A_{\d_1}}d(x,A)^{s-N}\D x.
\end{equation}
Note that $\d_1\le d(x,A)<\d_2$ for every $x\in A_{\d_2}\setminus A_{\d_1}$.
Hence, 
 the integral given by (\ref{int12}) is an entire function of the variable~$s$.
\end{proof}

The following result deals with the scaling property of the distance zeta function. 

\begin{prop}[Scaling property of distance zeta functions]\label{scalingd}
For any bounded subset $A$ of $\eR^N$, $\d>0$ and $\g>0$, we have $D(\zeta_{\g A}(\,\cdot\,,\g(A_\d)))=D(\zeta_A(\,\cdot\,,A_\d))=\ov\dim_BA$ and
\begin{equation}\label{zetalA}
\zeta_{\g A}(s,\g(A_\d))=\g^s\zeta_A(s,A_\d),
\end{equation}
for all $s\in\Ce$ with $\re s>\ov\dim_BA$. Furthermore, if $\o\in\Ce$ is a simple pole of the meromorphic extension of $\zeta_A(s,A_\d)$ to some open connected neighborhood of the critical line $\{\re s=\ov\dim_BA\}$ $($we use the same notation for the meromophically extended function$)$, then
\begin{equation}\label{reslA}
\res(\zeta_{\g A}(\,\cdot\,,\g(A_\d)),\o)=\g^{\o}\res(\zeta_A,\o).
\end{equation}
\end{prop}

\begin{proof}  
Equation \eqref{zetalA} follows easily by noting that $\g(A_\d)=(\g A)_{\g\d}$; we leave the details to the interested reader. 
To prove Equation \eqref{reslA}, note that 
by using \eqref{zetalA}, we obtain that
$$
\begin{aligned}
\res(\zeta_{\g A}(\,\cdot\,,\g(A_\d)),\o)&=\lim_{s\to \o}(s-\o)\zeta_{\g A}(s,\g A)\\
&=\lim_{s\to \o}(s-\o)\g^s\zeta_A(s,A)=\g^{\o}\res(\zeta_A,\o),
\end{aligned}
$$
which concludes the proof of the proposition.
\end{proof}

This scaling result is useful, in particular, in the study of fractal sprays and self-similar sets in Euclidean spaces; see [LapRa\v Zu3,5].

\section{Residues of zeta functions and Minkowski contents}\label{residues_m}

In this section, we show that the residue of any suitable meromorphic extension of the distance zeta function $\zeta_A$ of a fractal set $A$ in $\eR^N$ is closely
related to the Minkowski content of the set; see Theorems~\ref{pole1} and~\ref{pole1mink_tilde}. 
Therefore, the distance zeta functions, as well as the tube zeta functions that we introduce below (see Definition \ref{zeta_tilde}), can be considered as a useful tool in the study of the geometric properties of fractals. 

\subsection{Distance zeta functions of fractal sets and their residues}\label{residues_m_distance}
Here we use the notation $\zeta_A(s,A_\delta)$ for the distance zeta function
instead of $\zeta_A(s)$, in order to 
stress the dependence of the zeta function on~$\delta$. We start with an identity or functional equation, which will motivate us to introduce a new 
class of zeta functions, described by~(\ref{zeta_tilde}).

\begin{theorem}\label{equr}
Let $A$ be a bounded subset of $\eR^N$, and let $\delta$ be a fixed positive number. Then, for all $s\in\Ce$ such that $\re s>\ov\dim_BA$, 
the following identity holds$:$
\begin{equation}\label{equality}
\int_{A_\delta}d(x,A)^{s-N}\D x=\delta^{s-N}|A_\delta|+(N-s)\int_0^\delta t^{s-N-1}|A_t|\,\D t.
\end{equation}
Furthermore, the function  $\tilde\zeta_A(s):=\int_0^\delta t^{s-N-1}|A_t|\,\D t$ is absolutely convergent $($and hence, holomorphic$)$ on $\{\re s>\ov\dim_BA\}$. The function $\tilde\zeta_A$, which we have just introduced, is called the tube zeta function of $A$ $($see Definition \ref{zeta_tilde}$)$ and will be studied in Subsection \ref{residues_m_tube}.
\end{theorem}

\begin{proof}
Equality (\ref{equality}) holds for all real numbers $s\in(\ov D,+\ty)$, where $\ov D:=\ov\dim_BA$. Indeed, it follows immedately from Lemma~\ref{identity0}, if we take $\gamma:=N-s$ (note that then $\gamma<N-\ov D$).

Let us denote the left-hand side of (\ref{equality})
by $f(s)$, and the right-hand side by $g(s)$. Since $f(s)=g(s)$ on the subset $(\ov D,+\ty)\st\Ce$, to prove the theorem, it suffices to show that $f(s)$ and $g(s)$
are both holomorphic in the region $\{\re s>\ov D\}$. Indeed, the fact that (\ref{equality}) then holds for all $s\in\Ce$ with $\re s>\overline{D}$  follows from the principle of analytic continuation; see, e.g.,  \cite[Corollary~3.8]{conway}. The holomorphicity of $f(s)$ in that region is precisely the content of Theorem~\ref{an}$(a)$.

In order to prove the holomorphicity of $g(s)$ on $\{\re s>\ov D\}$, it suffices to show that $\tilde\zeta_A(s)$ is absolutely convergent on $\{\re s>\ov\dim_BA\}$. Note that $\tilde\zeta_A(s)$ is the Dirichlet-type integral, $\tilde\zeta_A(s)=\int_E\f(t)^s\D \mu(x)$, where $E:=(0,\delta)$, $\f(t):=t$, $d\mu(x):=t^{-N-1}|A_t|\,dt$, and the latter measure is positive. Therefore, it 
suffices to show that for any $s\in\Ce$ such that $\re s>\ov D$, the Dirichlet-type integral $\tilde\zeta_A(s)$ is well defined. To see this, let $\e>0$ be small enough, so that $\re s>\ov D+\e$. Since $\M^{*(\ov D+\e)}(A)=0$, there exists $C_\delta>0$ such that $|A_t|\le 
C_\delta t^{N-\ov D-\e}$ for all $t\in(0,\delta]$. Then
\begin{equation}\nonumber
\begin{aligned}
|\tilde\zeta_A(s)|&\le\int_0^\delta t^{\re s-N-1}|A_t|\,\D t\\
&\le C_\delta\int_0^\delta t^{\re s-\ov D-\e-1}\D t=C_\delta\frac{\delta^{\re s-\ov D-\e}}{\re s-\ov D-\e}<\ty,
\end{aligned}
\end{equation}
which concludes the proof of the theorem.
\end{proof}

\begin{cor}\label{equr_c}
If $\ov\dim_BA<N$, then
\begin{equation}
D(\zeta_A)=D(\tilde\zeta_A)\q\mbox{and}\q D_{\rm hol}(\zeta_A)=D_{\rm hol}(\tilde\zeta_A).
\end{equation}
\end{cor}

\begin{proof}
This follows at once from Equation \eqref{equality} of Theorem \ref{equr} and from the definition of $D(f)$ and $D_{\rm hol}(f)$, for $f=\zeta_A$ or $f=\tilde\zeta_A$.
\end{proof}

\medskip

The following theorem is, in particular, a higher-dimensional generalization of \cite[Theorem~1.17]{lapidusfrank12} and yields more information than the latter result, when $N=1$. (The problem of constructing meromorphic extensions
of fractal zeta functions is studied in \cite{mezf}.)

\begin{theorem}\label{pole1}
Assume that the bounded set $A\st\eR^N$ is Minkowski nondegenerate $($that is, $0<\M_*^D(A)\le\M^{*D}(A)<\ty$, and, in particular, $\dim_BA=D$$)$, and $D<N$. If, in addition, $\zeta_A(\,\cdot\,,A_\delta)$ can be extended meromorphically to a neighborhood of $s= D$,
then $D$ is necessarily a simple pole of $\zeta_A(\,\cdot\,,A_\delta)$, and 
the value of the residue of $\zeta_A(\,\cdot\,,A_\d)$ at $D$, $\res(\zeta_A(\,\cdot\,,A_\delta), D)$, does not depend on $\delta>0$. Furthermore,
\begin{equation}\label{res}
(N-D)\M_*^D(A)\le\res(\zeta_A(\,\cdot\,,A_\delta),D)\le(N-D)\M^{*D}(A),
\end{equation}
and
in particular, if $A$ is Minkowski measurable, then 
\begin{equation}\label{pole1minkg1=}
\res(\zeta_A(\,\cdot\,,A_\delta), D)=(N-D)\M^D(A).
\end{equation}
\end{theorem}

\begin{proof} Since $\M_*^D(A)>0$, using Theorem~\ref{an}(c) we conclude that $s=D$ is a pole of $\zeta_A=\zeta_A(\,\cdot\,,A_\d)$. Therefore, it suffices to show that the order of the pole at $s=D$ is not larger than $1$.
Let us take any fixed $\delta>0$, and let 
\begin{equation}\label{Cdelta}
C_\delta:=\sup_{t\in(0,\delta]}\frac{|A_t|}{t^{N-D}}.
\end{equation} 
Note that $C_\delta<\ty$ because $\M^{*D}(A)<\ty$. Then, in light of \eqref{equality}, for all $s\in\eR$ with $D<s<N$, we have
\begin{equation}\label{res0}
\begin{aligned}
\zeta_A(s,A_\delta)&=\int_{A_\delta}d(x,A)^{s-N}\D x=\delta^{s-N}|A_\delta|+(N-s)\int_0^\delta t^{s-N-1}|A_t|\,\D t\\
&\le C_\delta\delta^{s-D}+C_\delta(N-s)\frac{\delta^{s-D}}{s-D}=C_\delta(N-D)\delta^{s-D}\frac1{s-D}.
\end{aligned}
\end{equation}
Therefore, $0<\zeta_A(s,A_\delta)\leq C_1(s-D)^{-1}$ for all $s\in(D,N)$. This shows that $s=D$ is a pole of $\zeta_A(s,A_\delta)$ which is at most of order $1$, and the first claim is established. Namely, $D$ is a simple pole of $\zeta_A(s,A_{\delta})$.

The fact that the residue of $\zeta_A(s,A_\delta)$ at $s=D$ is independent of the value of $\delta>0$ follows immediately from Proposition \ref{O}.
In order to prove the second inequality in (\ref{res}), is suffices to multiply (\ref{res0}) by $s-D$, with $s$ real, and take the limit as $s\to D^+$ along the real axis:
\begin{equation}\label{res-delta}
\res(\zeta_A(\,\cdot\,,A_\delta),D)\le(N-D)\lim_{s\to D^+}C_\delta\delta^{s-D}=(N-D)C_\delta.
\end{equation} 
Since the residue of $\zeta_A(s,A_\delta)$ at $D$ does not depend on $\delta$,
 (\ref{res}) follows from (\ref{res-delta}) by recalling the definition of $C_\delta$ given in (\ref{Cdelta}) and passing to the limit as $\delta\to0^+$
 (note that the function $\d\mapsto C_\d$ is nondecreasing and that $C_\d\to\M^{*D}(A)$ as $\delta\to0^+$) on the right-hand side of (\ref{res-delta}). The first inequality in (\ref{res}) is proved analogously by replacing the supremum by an infimum in the definition of $C_\d$ given in \eqref{Cdelta}.
\end{proof}

\begin{example}[Residues of the zeta function of the generalized Cantor set]\label{res-cantor} 
Let $A=C^{(a)}$ be the generalized Cantor set 
 defined by the parameter $a\in(0,1/2)$. Recall that $C^{(a)}$ is obtained by deleting the middle interval of length $1-2a$ from the interval $[0,1]$,
and then continuing in the usual way, scaling by the factor $a$ at each step. 
(For $a=1/3$, we obtain the middle third Cantor set, which is studied in detail in \cite{lapiduspom} and, from the point of view of geometric zeta functions and the associated complex dimensions, in \cite{lapidusfrank12}.) 
By a direct computation, we obtain the corresponding zeta function:
\begin{equation}\label{cantor_z}
\zeta_A(s,A_\delta):=\frac{2^{1-s}(1-2a)^s}{s(1-2a^s)}+2\delta^ss^{-1}.
\end{equation}
Its residue computed at $D=D(a):=\dim_BA=\log_{1/a}2$ is given by
\begin{equation}\label{cantor_res}
\res(\zeta_A(\,\cdot\,,A_\delta),D)=\frac{2}{\log 2}\left(\frac12-a\right)^{D}.
\end{equation}
On the other hand, the values of the lower and upper $D$-dimensional Minkowski contents are respectively equal to (see~\cite[Equations (3.12) and (3.13) for $m=2$]{mink}):
\begin{equation}\label{cantorM}
\M_*^D(A)=\frac 1D\left(\frac{2D}{1-D}\right)^{1-D},\quad \M^{*D}(A)=2(1-a)\left(\frac12-a\right)^{D-1},
\end{equation} 
and thus $\M_*^D(A)<\M^{*D}(A)$ (see also Remark \ref{<} below). It follows that $C^{(a)}$ is not Minkowski measurable.
Therefore, for any generalized Cantor set $A=C^{(a)}$, with $a\in(0,1/2)$, we have that 
\begin{equation}\label{cantorM1}
(1-D)\M_*^D(A)<\res(\zeta_A(\,\cdot\,,A_\delta),D)<(1-D)\M^{*D}(A).
\end{equation} 
 This is in agreement with (\ref{res}) in Theorem~\ref{pole1}. In particular, since the functions $(0,1/2)\ni a\mapsto \M_*^D(A)$ and $a\mapsto \M^{*D}(A)$ are bounded, and
$D=\log_{1/a}2\to 1^-$ as $a\to1/2^-$, we have that for any positive $\delta$,
\begin{equation}
\lim_{a\to1/2^-}\res(\zeta_A(\,\cdot\,,A_\delta),D)=0.\nonumber
\end{equation}


The residues of $\zeta_A(s,A_\delta)$ at the poles $s_k:=D+k\mathbf{p}{\I}$, $k\in\Ze$, on the critical line $\{\re s=D\}$, expressed in terms of the residue at $D$ and the `oscillatory period' (see \cite{lapidusfrank12}) $\mathbf{p}:=2\pi/\log(1/a)$, are the following:
\begin{equation}
\res(\zeta_A(\,\cdot\,,A_\delta),s_k)=\frac{D2^{-k\mathbf{p}{\I}}(1-2a)^{k\mathbf{p}{\I}}}
{s_ka^{k\mathbf{p}{\I}}}\res(\zeta_A(\,\cdot\,,A_\delta),D),\quad k\in\Ze.
\end{equation}
\end{example}

\begin{remark}\label{<}
As we have already noted, the two inequalities in (\ref{cantorM1}) are in agreement with (\ref{res}) in Theorem~\ref{pole1}.
In \cite{mezf}, we prove that the strict inequalities in (\ref{res}) are not just a coincidence: indeed, they hold
for a large class of Minkowski nonmeasurable sets in Euclidean spaces. An analogous remark applies to 
the inequalities (\ref{zeta_tilde_M}) in Theorem \ref{pole1mink_tilde} below, dealing with tube zeta functions.
\end{remark}

\subsection{Tube zeta functions of fractal sets and their residues}\label{residues_m_tube}
Going back to Theorem~\ref{equr}, we see that it is natural to introduce a new fractal zeta function of bounded subsets $A$ of $\eR^N$.

\begin{defn}\label{zeta_tilde}  Let $\delta$ be a fixed positive number, and let $A$ be a bounded subset of $\eR^N$. Then, the {\em tube zeta function} of $A$, denoted by $\tilde\zeta_A$, is defined by
\begin{equation}\label{tildz}
\tilde\zeta_A(s)=\int_0^\delta t^{s-N-1}|A_t|\,\D t,
\end{equation}
for all $s\in\Ce$ with $\re s$ sufficiently large. As we know from Theorem \ref{equr}, the tube zeta function is (absolutely) convergent
(and hence, holomorphic) on the open right half-plane $\{\re s>\ov\dim_BA\}$.
\end{defn} 

We call $\tilde\zeta_A$ the tube zeta function of $A$ since its definition involves the tube function $(0,\delta)\ni t\mapsto |A_t|$. Relation (\ref{equality}) can be written as follows (with $\zeta_A(s)=\zeta_A(s,A_\d)$, as before, and $\tilde\zeta_A(s)=\tilde\zeta_A(s,A_\d)$, for emphasis):
\begin{equation}\label{equ_tilde}
\zeta_A(s,A_\delta)=\delta^{s-N}|A_\delta|+(N-s)\tilde\zeta_A(s,A_\d),
\end{equation}
for any $\delta>0$ and for all $s\in\Ce$ such that $\re s>\ov\dim_BA$.


From the {\em functional equation} \eqref{equ_tilde} relating $\zeta_A$ and $\tilde\zeta_A$, it would seem that $\tilde\zeta_A$ has a singularity at $s=N$.
However, from the second part of Theorem \ref{equr} we see that for $\ov\dim_BA<N$, the value $s=N$ is regular (i.e., holomorphic) for $\tilde\zeta_A$.
It then follows from \eqref{equ_tilde} that the two fractal zeta functions $\zeta_A$ and $\tilde\zeta_A$ contain essentially the same information.

In particular, still assuming that $\ov\dim_BA<N$, $\tilde\zeta_A$ has a meromorphic continuation to a given domain $U\stq\Ce$ (containing the critical line $\{\re s=\ov\dim_BA\}$) if and only if $\zeta_A$ does, and in that case (according to the principle of analytic continuation), the unique meromorphic continuations to $U$ of $\zeta_A$ and $\tilde\zeta_A$ are still related by the functional equation \eqref{equ_tilde}.
Also in that case, the residues (or, more generally, the principal parts) of $\zeta_A$ and $\tilde\zeta_A$ of a given simple (resp., multiple) pole of $s=\o\in U$ are related in a very simple manner; see, e.g., Equation \eqref{1.3.18} below in the case of the simple pole $s=\ov\dim_BA$. Furthermore, $\mathcal{P}(\zeta_A)=\mathcal{P}(\tilde\zeta_A)$ and (assuming that $U$ contains the critical line $\{\re s=\ov\dim_BA\}$), $\mathcal{P}_c(\zeta_A)=\mathcal{P}_c(\tilde\zeta_A)$.

Moreover, we have that $D(\tilde\zeta_A)=D(\zeta_A)$, $D_{\rm hol}(\tilde\zeta_A)=D_{\rm hol}(\zeta_A)$ and $D_{\rm mer}(\tilde\zeta_A)=D_{\rm mer}(\zeta_A)$.
(Here, $D_{\rm mer}(f)$, the {\em abscissa of meromorphic continuation} of a given meromorphic function $f$, is defined exactly as $D_{\rm hol}(f)$ in Equation \eqref{Dhol} and the surrounding text, except for ``holomorphic'' replaced by ``meromorphic''; and similarly for the half-plane of meromorphic continuation of $f$.)
Also, we have $\Pi(\tilde\zeta_A)=\Pi(\zeta_A)$ and $\mathcal{H}(\tilde\zeta_A)=\mathcal{H}(\zeta_A)$; similarly, the half-planes of meromorphic continuation of $\tilde\zeta_A$ and $\zeta_A$ coincide.

Still in light of \eqref{equ_tilde}, it follows from Theorem \ref{equr} that $\tilde\zeta_A$ is holomorphic on $\{\re s>\ov\dim_BA\}$ 
and that (provided $\ov\dim_BA<N$), the lower bound $\ov\dim_BA$ is optimal from the point of view of the convergence of the Lebesgue integral defining $\zeta_A$
in \eqref{tildz}; i.e., $D(\tilde\zeta_A)\,(=D(\zeta_A))=\ov\dim_BA$. More generally, the exact analog of Theorem \ref{an} holds for $\tilde\zeta_A$ (instead of $\zeta_A$), except for the fact that in the counterpart of part $(c)$ of Theorem  \ref{an} we no longer need to assume that $D<N$ (where $D:=\dim_BA$).

Assuming that there exists a meromorphic extension of $\zeta_A(s,A_\delta)$
to an open connected neighborhood of $\ov D:=\ov\dim_BA$, and $\ov D$ is a simple pole, $\ov D<N$, then it easily follows from (\ref{equ_tilde}) that
\begin{equation}\label{1.3.18}
\res(\tilde\zeta_A,\ov D)=\frac1{N-\ov D}\res(\zeta_A(\,\cdot\,,A_\delta),\ov D).
\end{equation}
Indeed, 
\begin{eqnarray}
\res(\zeta_A(\,\cdot\,,A_\delta),\ov D)&=&\lim_{s\to \ov D}(s-\ov D)[\delta^{s-N}|A_\delta|+(N-s)\tilde\zeta_A(s)]\nonumber\\
&=&(N-\ov D)\lim_{s\to\ov D}(s-\ov D)\tilde\zeta_A(s)\nonumber\\
&=&(N-\ov D)\res(\tilde\zeta_A,\ov D).\nonumber
\end{eqnarray}
Hence, the following result, in the case when $D<N$, is an immediate consequence of Theorem~\ref{pole1} and relation (\ref{equality}) (or, equivalently, \eqref{equ_tilde}), while in the case when $D=N$, it can be shown directly.

\begin{theorem}\label{pole1mink_tilde}
Assume that $A$ is a bounded subset of $\eR^N$ such that  $D:=\dim_BA$ exists, 
 $0<\M_*^D(A)\le\M^{*D}(A)<\ty$,
and there exists a meromorphic extension of $\tilde\zeta_A$ to an open neighborhood of $D$.
Then $D$ is a simple pole, and for any positive $\delta$, the value of $\res(\tilde{\zeta}_A,D)$ is independent of $\delta$.
Furthermore, we have
\begin{equation}\label{zeta_tilde_M}
\M_*^D(A)\le\res(\tilde\zeta_A, D)\le \M^{*D}(A),
\end{equation}
and, in particular, if $A$ is Minkowski measurable, then 
\begin{equation}\label{zeta_tilde_Mm}
\res(\tilde\zeta_A, D)=\M^D(A).
\end{equation}
\end{theorem}

In the following example, we compute the complex dimensions of the unit $(N-1)$-dimensi\-o\-nal sphere in $\eR^N$, using the tube zeta function of the sphere.

\begin{example}\label{sphere} Let $A:=\pa B_1(0)$ be the unit $(N-1)$-dimensional sphere in $\eR^N$ centered at the origin. We would like to compute its complex dimensions. To this end, we first compute the corresponding tube zeta function $\tilde\zeta_A$. Let us fix any $\d\in(0,1)$. Since $|A_t|=\o_N(1+t)^N-\o_N(1-t)^N$, where $t\in(0,1)$ and $\o_N$ is the $N$-dimensional Lebesgue measure of the unit ball in $\eR^N$, we have that for any fixed $\d\in(0,1)$,
$$
\begin{aligned}
\tilde\zeta_A(s)&=\int_0^\d t^{s-N-1}|A_t|\,\D t=\o_N\int_0^\d t^{s-N-1}((1+t)^N-(1-t)^N)\,\D t\\
&=\o_N\int_0^\d t^{s-N-1}\Bigg(\sum_{k=0}^N\binom Nk\big(1-(-1)^k\big) t^k\Bigg)\,\D t\\
&=\o_N\sum_{k=1}^N\big(1-(-1)^k\big) \binom Nk\frac{\d^{s-N+k}}{s-(N-k)},
\end{aligned}
$$
for all $s\in\Ce$ with $\re s>N-1$.
The last expression can be meromorphically extended to the whole complex plane, and we still denote it by $\tilde\zeta_A(s)$. Therefore,
we have
\begin{equation}\label{SN-1}
\tilde\zeta_A(s)=\o_N\sum_{k=0}^N\big(1-(-1)^k\big)\binom Nk\frac{\d^{s-N+k}}{s-(N-k)},
\end{equation}
for all $s\in\Ce$. It follows that
\begin{equation}\label{dimBAN-1}
\begin{gathered}
\dim_BA=D(\tilde\zeta_A)=D(\zeta_A)=N-1,\\
\po_c(\tilde\zeta_A)=\po_c(\zeta_A)=\{N-1\},
\end{gathered}
\end{equation}
as expected. (Note that $\dim_BA=N-1<N$, so that $\po_c(\tilde\zeta_A)=\po_c(\zeta_A)$ and $\po(\tilde\zeta_A)=\po(\zeta_A)$.)
Moreover, still in light of \eqref{SN-1}, the set of complex dimensions of $A$ is given by (with $\lfloor x\rfloor$ denoting the integer part of $x\in\eR$)
\begin{equation}\label{dimSN-1}
\begin{aligned}
\po(\tilde\zeta_A)=\po(\zeta_A)&=\Big\{N-(2j+1):j=0,1,2,\dots,\Big\lfloor\frac{N-1}2\Big\rfloor\Big\}\\
&=\Big\{N-1,N-3,\dots,N-\Big(2\Big\lfloor\frac{N-1}2\Big\rfloor+1\Big)\Big\}.
\end{aligned}
\end{equation}
For odd $N$, the last number in this set is equal to $0$, while for even $N$, it is equal to $1$.
Furthermore, the residue of the tube zeta function $\tilde\zeta_A$ at any of its poles $N-k\in\po(\tilde\zeta_A)$ is given by
$\res(\tilde\zeta_A,N-k)=2\o_N\binom Nk$;
that is,
\begin{equation}\label{resSN-1b}
\res(\tilde\zeta_A,d)=2\o_N\binom Nd,\q\mbox{for all\q $d\in\po(\tilde\zeta_A)$}.
\end{equation}
Note that in the case when $d=D:=N-1$, we obtain
\begin{equation}\label{resBR0}
\res(\tilde\zeta_A,D)=2N\o_N=\M^D(A),
\end{equation}
where the last equality is easily obtained from the definition of the Minkowski content, as follows:
$$
\M^D(A)=\lim_{t\to0^+}\frac{|A_t|}{t^{N-D}}=\lim_{t\to0^+}\frac{\o_N(1+t)^N-\o_N(1-t)^N}{t}=2N\o_N.
$$
In other words, $A$ is Minkowski measurable and
\begin{equation}\label{MHD}
\M^D(A)=2\, \mathcal{H}^D(A),
\end{equation}
where $\mathcal{H}^D$ denotes the $D$-dimensional Hausdorff measure. (Equation \eqref{MHD} is a special case of a much more general result proved by Federer in \cite[Theorem 3.2.39]{federer}.)
Equation \eqref{resBR0} is in agreement with Equation \eqref{zeta_tilde_Mm} in Theorem \ref{pole1mink_tilde}.
\end{example}

\subsection{Residues of tube zeta functions of generalized Cantor sets and $a$-strings}\label{residues_m_zeta}
We provide here two simple examples illustrating some of the main results of this section.

\begin{example}[Generalized Cantor sets, Example~\ref{res-cantor} continued]\label{res-cantor2}
As an illustration of inequality (\ref{zeta_tilde_M}), we consider generalized Cantors sets, $A=C^{(a)}$, $a\in(0,1/2)$. We obtain
\begin{equation}\label{cantorM2}
\M_*^D(A)<\res(\tilde\zeta_A(\,\cdot\,,A_\delta),D)<\M^{*D}(A),
\end{equation} 
where the values of the lower and upper Minkowski contents, $\M_*^D(A)$ and $\M^{*D}(A)$, are given by \eqref{cantorM} and $D=D(a)=\log_{1/a}2$.
It is worth observing that $C^{(a)}$ becomes almost like a Minkowski measurable set for $a$ close to $1/2$, since
both $\M^{*D}(A)$ and $\M_*^D(A)$ tend to the common limit $1$ as $a\to1/2^-$. 

On the other hand, in the limit where $a\to0^+$, $C^{(a)}$ remains Minkowski nonmeasurable since
\begin{equation}
\lim_{a\to0^+}\M^{*D}(A)=4,\quad \lim_{a\to0^+}\M_*^D(A)=2.
\end{equation}
\end{example}

\begin{example}[$a$-strings]\label{a-string2}
Given $a>0$, the associated {\em $a$-string} is defined by $\mathcal L=(\ell_j)_{j\ge1}$,
where $\ell_j=j^{-a}-(j+1)^{-a}$. Let $A=A_{\mathcal L}=\{j^{-a}:j\in\eN\}$ be the associated set; see Example \ref{L} and the discussion preceding it. This set is Minkowski measurable,
\begin{equation}\label{a-string}
\M^D(A)=\frac{2^{1-D}}{D(1-D)}a^D,\quad D=D(a)=\frac1{1+a}.
\end{equation} 
This fractal string has been introduced in \cite[Example 5.1]{Lap1}.
Due to (\ref{pole1minkg1=}) and (\ref{zeta_tilde_Mm}), we know that
\begin{equation}\label{a-string-mink}
\res(\zeta_A(\,\cdot\,,A_\delta),D)=(1-D)\M^D(A),\quad \res(\tilde\zeta_A,D)=\M^D(A).
\end{equation}
\end{example}

\subsection{Distance and tube zeta functions of fractal grills}\label{dtx}

It is of interest to understand the behavior of the distance and tube zeta functions with respect to the Cartesian products of sets.
In this subsection, we restrict our attention to Cartesian products of the form $A\times[0,1]^k\st\eR^{N+k}$, which we call {\em fractal grills}. Here, $A$ is a bounded subset of $\eR^N$ and $k$ is any positive integer.

Since the set $A$ can be naturally identified with $A\times\{0\}\st\eR^{N+1}$, it will be convenient to introduce the following notation for all $s\in\Ce$ with $\re s$ sufficiently large:
\begin{equation}\label{[N]}
\zeta_A^{[N]}(s):=\int_{A_\d}d(x,A)^{s-N}\,\D x,\q
\tilde\zeta_A^{[N]}(s):=\int_0^\d t^{s-N-1}|A_t|_N\D t,
\end{equation} 
where the index $[N]$ indicates that we view $A$ as a subset of $\eR^N$ and $|A_t|_N$ is the $N$-dimensional Lebesgue measure of the $t$-neighborhood of $A$ in $\eR^N$. Hence, $\tilde\zeta_A^{[N+1]}(s)=\int_0^\d t^{s-N-2}|(A\times\{0\})_t|_{N+1}\D t$. Note that, by writing $|(A\times\{0\})_t|_{N+1}$, we  interpret $(A\times\{0\})_t$ as the $t$-neighborhood of $A\times\{0\}$ in $\eR^{N+1}$. 
Furthermore, observe that, in \eqref{[N]}, $\zeta_A^{[N]}$ and $\tilde\zeta_A^{[N]}$, are, respectively, the usual distance and tube zeta functions of $A$ (viewed as a bounded subset of $\eR^N$) whereas, for example, $\tilde\zeta_A^{[N+1]}$ is the tube zeta function of $A\times\{0\}$, but now viewed instead as a subset of $\eR^{N+1}$. Moreover, in \eqref{1} and \eqref{2} of Lemma \ref{cartesian} just below, $\zeta_{A\times[0,1]}^{[N+1]}$ and $\tilde\zeta_{A\times[0,1]}^{[N+1]}$ stand, respectively, for the usual distance and tube zeta functions of $A\times[0,1]$ (naturally viewed as a subset of $\eR^{N+1}$).

In the sequel, if $\Sigma$ is a given set of complex numbers and $\kappa\in\Ce$ a fixed complex number, we let $\Sigma+\kappa:=\{s+\kappa:s\in \Sigma\}$.
We shall also need the following definition.

\begin{defn}\label{simeq1}
Assume that $f(s)$ and $g(s)$ are two tamed Dirichlet-type integrals (DTIs, in short) which are (absolutely) convergent on an open right half-plane $\{\re s>\a\}$, for some $\a\in\eR$. Let their difference $h(s):=f(s)-g(s)$ be a tamed DTI
such that
$D(h)<D(g)$. (Or, equivalently, that there exists a real number $\b$, with $\b<D(g)$, such that the integral defining $h$ is absolutely convergent (and hence, holomorphic) on $\{\re s>\b\}$.) Then we say that $f$ and $g$ are {\em weakly equivalent} and write $f\simeq g$. 
\end{defn}

\begin{remark}
It can be checked that if $f$ and $g$ are tamed DTIs, then $f-g$ (or, more generally, any linear combination of $f$ and $g$) is a tamed DTI (as is required in Definition \ref{simeq1} just above) provided both the DTIs $f$ and $g$ are based on the same underlying pair $(E,\f)$ in the notation of Definition \ref{abscissa_f}. 
Therefore, $D(h)$ and $\Pi(h)$ are well defined in that case.
This situation arises, for example, for the tube zeta function discussed in the present subsection. We then have $E:=(0,\d)$ and $\f(t):=t$ for all $t\in E$.
\end{remark}

Note that in Definition \ref{simeq1}, we do not assume that $g$ possesses a meromorphic continuation to a neighborhood of any point on its critical line
$\{\re s=D(g)\}$. 
Case $(c)$ of Lemma \ref{simeq2} below provides a simple and useful condition for the implication $f\simeq g$ $\implies$ $f\sim g$ to hold, where the equivalence $\sim$ is described in Definition \ref{equ} above.

\begin{lemma}\label{simeq2}
Assume that $f$ and $g$ are two tamed Dirichlet-type integrals such that $f\simeq g$.
Then, the following properties hold$:$
\medskip

$(a)$ We have $D(f)=D(g)$. 
\smallskip

$(b)$ The relation $\simeq$ is reflexive and symmetric.
\smallskip

$(c)$ If there exists a connected open set $U\stq\{\re s>D(f-g)\}$ containing the critical line $\{\re s = D(g)\}$ and such that
$g$ can be meromorphically continued to $U$, then $f$ has the same property and $\po_c(f)=\po_c(g)$.
In particular, $f\sim g$ in the sense of Definition \ref{equ}.
\end{lemma}

\begin{proof}
$(a)$ Since, by Definition \ref{simeq1}, $f(s)=g(s)+h(s)$ and $D(h)<D(g)$,  we conclude that $D(f)\le D(g)$.
If we had $D(f)<D(g)$, then we would have 
\begin{equation}\label{contr}
\max\{D(f),D(h)\}<D(g).
\end{equation}
On the other hand,
the function (i.e., the DTI) $g(s)=f(s)-h(s)$ is absolutely convergent on $\{\re s>\max\{D(f),D(h)\}\}$,
which is impossible due to \eqref{contr}. This contradiction proves that $D(f)=D(g)$.

 Property $(b)$ follows at once from $(a)$ and Definition \ref{simeq1}.  Finally, property $(c)$ follows easily from the relation $f(s)=g(s)+h(s)$.
\end{proof}

\begin{lemma}\label{cartesian}
Let $A$ be a bounded subset of $\eR^N$. Then 
\begin{equation}\label{1}
\zeta_{A\times[0,1]}^{[N+1]}(s)=\zeta_A^{[N]}(s-1)+\zeta_A^{[N+1]}(s)
\end{equation}
and
\begin{equation}\label{2}
\tilde\zeta_{A\times[0,1]}^{[N+1]}(s)=\tilde\zeta_A^{[N]}(s-1)+\tilde\zeta_A^{[N+1]}(s)
\end{equation}
for all $s\in\Ce$ with $\re s>\ov\dim_BA+1$.
In particular, if $A$ is such that $\zeta_A$ or $($equivalently, provided $\ov\dim_BA<N$$)$ $\tilde\zeta_A$ admits a $($necessarily unique$)$ meromorphic continuation to a
connected open neighborhood
of the critical line of Lebesgue $($absolute$)$ convergence $\{\re s=D(\zeta_A)\}$ $($recall from Theorem \ref{an} that $D(\zeta_A)=\ov\dim_BA$$)$, then
\begin{equation}
\zeta_{A\times[0,1]}^{[N+1]}(s)\simeq\zeta_A^{[N]}(s-1)\q\mbox{\rm and}\q
\tilde\zeta_{A\times[0,1]}^{[N+1]}(s)\simeq\tilde\zeta_A^{[N]}(s-1).
\end{equation}
Hence, if $\zeta_A$ can be meromorphically continued to a connected, open set $U$ containing the critical line
$\{\re s=D(\zeta_A)\}$, then $\mathcal{P}_c(\zeta_{A\times[0,1]}^{[N+1]})=\mathcal{P}_c(\zeta_A^{[N]})+1$; that is,
\begin{equation}\dim_{PC}(A\times[0,1])=\dim_{PC} A+1.
\end{equation} 
In particular, if $\ov\dim_BA<N$, then
\begin{equation}
\begin{aligned}
D(\zeta_{A\times[0,1]}^{[N+1]})&=D(\zeta_A^{[N]})+1=D(\tilde\zeta_A^{[N]})+1=D(\tilde\zeta_{A\times[0,1]}^{[N+1]})\\
&=\ov\dim_B(A\times[0,1])=\ov\dim_BA+1.
\end{aligned}
\end{equation}
\end{lemma}

\begin{proof} Let us first prove Equation \eqref{2}.
It is easy to see (cf.\ \cite[Remark 1]{maja}) that:
\begin{equation}\label{Aid}
|(A\times[0,1])_t|_{N+1}=|A_t|_N\cdot 1+|(A\times\{0\})_t|_{N+1}.
\end{equation}
Substituting into the second equality of \eqref{[N]}, we conclude that
\begin{equation}\label{m=1}
\begin{aligned}
\tilde\zeta_{A\times[0,1]}^{[N+1]}(s)&=\int_0^\d t^{s-N-2}(|A_t|_N+|(A\times\{0\})_t|_{N+1})\,\D t\\
&=\int_0^\d t^{(s-1)-N-1}|A_t|_N\D t+
\int_0^\d t^{s-(N+1)-1}|(A\times\{0\})_t|_{N+1}\D t\\
&=\tilde\zeta_A^{[N]}(s-1)+\tilde\zeta_A^{[N+1]}(s)
\end{aligned}
\end{equation}
\vskip1mm
\noindent for all $s\in\Ce$ with $\re s>\ov\dim_BA+1$.  (Here, we also use the fact that $\ov\dim_BA$ is the same in the case of $A\times\{0\}\st\eR^{N+1}$, as in the case of $A\st\eR^N$; that is, the upper box dimension of a bounded set, as well as the lower box dimension, does not depend on $N$; see [Kne, Satz~7] or \cite[Proposition~1]{maja}.)

Let us next establish Equation \eqref{1}. To this end, we use \eqref{equ_tilde}, which we write in the following form:
\begin{equation}
\tilde\zeta_A^{[N]}(s)=\frac{\zeta_A^{[N]}(s)-\d^{s-N}|A_\d|_N}{N-s},
\end{equation}
for all $s\in\Ce$ with $\re s>\ov\dim_BA$ and $s\ne N$.
\medskip
 Making use of Equation \eqref{m=1}, we deduce that
\begin{equation}
\begin{aligned}\label{Aida}
\frac{\zeta_{A\times[0,1]}^{[N+1]}(s)-\d^{s-N-1}|(A\times[0,1])_\d|_{N+1}}{(N+1)-s}&=\frac{\zeta_A^{[N]}(s-1)-\d^{(s-1)-N}|A_\d|_N}{N-(s-1)}\\
&\phantom{=}+\frac{\zeta_A^{[N+1]}(s)-\d^{s-(N+1)}|(A\times\{0\})_\d|_{N+1}}{(N+1)-s},
\end{aligned}
\end{equation}
for all $s\in\Ce$ with $\re s>\ov\dim_BA$ and $s\ne N+1$. Since, in light of \eqref{Aid}, we have $|(A\times[0,1])_\d|_{N+1}=|A_\d|_N+|(A\times\{0\})_\d|_{N+1}$,
 we conclude from \eqref{Aida} after a short computation that
\begin{equation}
\zeta_{A\times[0,1]}^{[N+1]}(s)=\zeta_A^{[N]}(s-1)+\zeta_A^{[N+1]}(s),
\end{equation}
 for all $s\in\Ce$ with $\re s>\ov\dim_BA+1$, where we have also used the principle of analytic continuation.
Note that, according to Theorem \ref{an}, both $\zeta_A^{[N]}(s-1)$ and $\zeta_{A\times[0,1]}^{[N+1]}(s)$ are holomorphic on $\{\re s>\ov\dim_BA+1\}$
(recall that $\ov\dim_B(A\times[0,1])=\ov\dim_BA+1$, see \cite{falc}), while, according to the same theorem, the function $\zeta_{A\times[0,1]}^{[N+1]}(s)-\zeta_A^{[N]}(s-1)=\zeta_A^{[N+1]}(s)$ is holomorphic on $\{\re s>\ov\dim_BA\}$. Therefore, since 
$D(\zeta_A^{[N+1]})=\ov\dim_BA<\ov\dim_BA+1=D(\zeta_A^{[N]}(\,\cdot\,-1))$, it follows from Definition \ref{simeq1} that
$\zeta_{A\times[0,1]}^{[N+1]}(s)\simeq\zeta_A^{[N]}(s-1)$.

The remaining part of Lemma \ref{cartesian} can be deduced from part $(c)$ of Lemma \ref{simeq2} by noting that since $\zeta_A(s)$ can be 
meromorphically continued to the set $U$, then $\zeta_A(s-1)$ can be meromorphically continued to the set $U+1$. Hence, by Lemma \ref{simeq2}$(c)$, we have 
$\zeta_{A\times[0,1]}^{[N+1]}(s)\sim\zeta_A^{[N]}(s-1)$ in the sense of Definition \ref{equ}, and therefore, 
$$
\po_c\big(\zeta_{A\times[0,1]}^{[N+1]}\big)=\po_c\big(\zeta_A^{[N]}(\,\cdot\,-1)\big)=\po_c\big(\zeta_A^{[N]}\big)+1,
$$ 
or, equivalently, $\dim_{PC}(A\times[0,1])=\dim_{PC}A+1$.
This completes the proof of the lemma.
\end{proof}

\begin{theorem}\label{Axm}
Let $A$ be a bounded subset of $\eR^N$ and let $d$ be a positive integer. Then the following properties hold: 
\medskip

$($a$)$ The distance and tube zeta functions of $A\times[0,1]^d\st\eR^{N+d}$ are given, respectively, by
\begin{equation}\label{md}
\zeta_{A\times[0,1]^d}^{[N+d]}(s)=\sum_{k=0}^d \binom dk \zeta_A^{[N+k]}(s-d+k)
\end{equation}
and
\begin{equation}\label{m}
\tilde\zeta_{A\times[0,1]^d}^{[N+d]}(s)=\sum_{k=0}^d \binom dk \tilde\zeta_A^{[N+k]}(s-d+k),
\end{equation}
for all $s\in\Ce$ with $\re s>\ov\dim_BA+d$. 

\medskip

$($b$)$ If the distance zeta function $\zeta_A$ $($or, equivalently, the tube zeta function $\tilde\zeta_A$$)$ can be meromophically extended to a connected open set containing the
critical line $\{\re s=\ov\dim_BA\}$, then
\begin{equation}\label{Axmsim}
\zeta_{A\times[0,1]^d}^{[N+d]}(s)\sim\zeta_A^{[N]}(s-d),\q
\tilde\zeta_{A\times[0,1]^d}^{[N+d]}(s)\sim\tilde\zeta_A^{[N]}(s-d)
\end{equation}
and  $\mathcal{P}_c(\zeta_{A\times[0,1]^d})=\mathcal{P}_c(\zeta_A)+d$; that is,
\begin{equation}\label{dimCm}
\dim_{PC}(A\times[0,1]^d)=\dim_{PC} A+d.
\end{equation}
In particular, if $\ov\dim_BA<N$, then
\begin{equation}
\begin{aligned}
D(\zeta_{A\times[0,1]^d}^{[N+d]})&=D(\zeta_A^{[N]})+d=D(\tilde\zeta_A^{[N]})+d=D(\tilde\zeta_{A\times[0,1]^d}^{[N+d]})\\
&=\ov\dim_B(A\times[0,1]^d)=\ov\dim_BA+d.
\end{aligned}
\end{equation}
\end{theorem}

\begin{proof}
$(a)$ Let us first prove Equation \eqref{md}. We do so by using mathematical induction on $d$. The case when $d=1$
has already been established in Lemma \ref{cartesian}.

Now, let us assume that the claim holds for some fixed positive integer $d\ge1$. From \eqref{1}
we see that
$$
\zeta_{A\times[0,1]^{d+1}}^{[N+d+1]}(s)=\zeta_{A\times[0,1]^{d}}^{[N+d]}(s-1)+\zeta_{A\times[0,1]^{d}}^{[(N+1)+d]}(s).
$$
Therefore,
\begin{equation}\nonumber
\begin{aligned}
\zeta_{A\times[0,1]^{d+1}}^{[N+d+1]}(s)&=\sum_{k=0}^d \binom dk \zeta_A^{[N+k]}(s-1-d+k)+
\sum_{k=0}^d \binom dk \tilde\zeta_A^{[N+1+k]}(s-d+k)\\
&=\zeta_A^{[N]}(s-d-1)+\sum_{k=0}^{d-1}\binom d{k+1}\zeta_A^{[N+k+1]}(s-d+k)\\
&\phantom{=}+\sum_{k=0}^{d-1} \binom dk \zeta_A^{[N+1+k]}(s-d+k)+\zeta_A^{[N+1+d]}(s)\\
&=\sum_{k=0}^{d+1} \binom {d+1}k \zeta_A^{[N+k]}(s-(d+1)+k),
\end{aligned}
\end{equation}
where in the last equality we have used the fact that $\binom dk+\binom d{k+1}=\binom{d+1}{k+1}$.
This completes the proof of Equation \eqref{md}.

Equation \eqref{m} can be proved by mathematical induction in much the same way as in the case of the distance zeta function. This completes the proof
 of part $(a)$ of the theorem.

\medskip

$(b)$ To prove that $\zeta_{A\times[0,1]^d}^{[N+d]}(s)\sim\zeta_A^{[N]}(s-d)$, it suffices to note that, by Equation \eqref{md}, the function
\begin{equation}
h(s):=\zeta_{A\times[0,1]^d}^{[N+d]}(s)-\zeta_A^{[N]}(s-d)=\sum_{k=1}^d \binom dk \zeta_A^{[N+k]}(s-d+k)
\end{equation}
has for abscissa of convergence $D(h)=\ov\dim_BA+(d-1)\}<\ov\dim_BA+d=D(\zeta_A^{[N]}(\,\cdot\,-d))$, so that 
$\zeta_{A\times[0,1]^d}^{[N+d]}(s)\simeq \zeta_A^{[N]}(s-d)$. Using part $(c)$ of Lemma \ref{simeq2}, we deduce that 
$\zeta_{A\times[0,1]^d}^{[N+d]}(s)\sim \zeta_A^{[N]}(s-d)$ in the sense of Definition~\ref{equ}, 
which proves the first relation in~\eqref{Axmsim}.
The second relation in \eqref{Axmsim} can be proved along the same lines.
This completes the proof of claim $(b)$, as well as of the entire theorem.
\end{proof}

\begin{remark}
The relations appearing in \eqref{Axmsim} can be written in a less precise form as follows:
\begin{equation}\label{Axmsim1}
\zeta_{A\times[0,1]^d}(s)\sim\zeta_A(s-d)\q\mbox{and}\q
\tilde\zeta_{A\times[0,1]^d}(s)\sim\tilde\zeta_A(s-d).
\end{equation}
We propose to call these two properties the {\em shift properties} of the distance and tube zeta functions, respectively.
\end{remark}

\begin{example}\label{Cmae}
Let $C^{(m,a)}$ be the two-parameter generalized Cantor set introduced in Definition \ref{Cma} below and let $d$ be a positive integer. Then, using \eqref{Axmsim} and \eqref{zetasim} below, we obtain that
$$
\zeta_{C^{(m,a)}\times[0,1]^d}(s)\sim\frac1{1-ma^{s-d}}.
$$
Furthermore, we conclude from \eqref{dimCm} that
\begin{equation}\label{CmaPC}
\dim_{PC}(C^{(m,a)}\times[0,1]^d)=(\log_{1/a}m+d)+\frac{2\pi}{\log (1/a)}\,\I\Ze.
\end{equation}
Moreover, by noticing that $\zeta_{C^{(m,a)}\times[0,1]^d}$ can be meromorphically extended to the whole complex plane, we conclude from Equation \eqref{md} above and from the first part of Equation \eqref{2.1.6} below
that the set of all complex dimensions of $C^{(m,a)}\times[0,1]^d\st\eR^{1+d}$ is well defined in $\Ce$ and given by
\begin{equation}
\po(\zeta_{C^{(m,a)}\times[0,1]^d})=\{0,1,\dots,d\}\cup\bigcup_{k=0}^d\Big((\log_{1/a}m+k)+\frac{2\pi}{\log (1/a)}\,\I\Ze\Big).
\end{equation}
The sets of the form $C^{(m,a)}\times[0,1]^d$ (with $m:=2$, $a:=1/3$, $d:=1$) appear, for example, in the study of the
Smale horseshoe map; see, e.g., \cite{smale}. They also arise in the study of the
 singularities of Sobolev functions and of weak solutions of elliptic equations; see, e.g., \cite{lana}, where they are called the `Cantor grills'.
\end{example}

\begin{example}\label{combs}
Similarly as in Example \ref{Cmae}, sets of the form $\pa\O\times[0,1]^{N-1}$, where $\O=\O_a$ is a geometric realization of a fractal string (for example, the so-called $a$-string, $\O=\cup_{j=1}^\ty((j+1)^{-a},j^{-a})$), where $a>0$ and for which $\pa\O=\{j^{-a}:j\ge1\}\cup\{0\}$ satisfies 
$\ov\dim_B\pa\O=1/(a+1)$, are used in the study of fractal drums to extend certain results from one to higher dimensions $N\ge2$; see \cite[Examples 5.1 and 5.1']{Lap1}.
The boundary of the open set $\O\times(0,1)^{N-1}$ is given by
\begin{equation} 
(\pa\O\times[0,1]^{N-1})\cup\big([0,1]\times\pa((0,1)^{N-1})\big),
\end{equation}
where $\pa\big(([0,1]^{N-1}\big)$ is taken in the space $\eR^{N-1}$.
The subset $\pa\big((0,1)^{N-1}\big)$ of $\eR^{N-1}$ is an $(N-2)$-dimensional Lipschitz surface (which for $N=2$ degenerates to a pair of points), so that the box dimension of $[0,1]\times\pa((0,1)^{N-1})$ is equal to $N-1$. Therefore, by the property of `finite stability' of the upper box dimension (see \cite{falc}), we have
$\ov\dim_B(\O\times(0,1)^{N-1})=\max\{\ov\dim_B(\pa\O\times[0,1]^{N-1}),N-1\}=\ov\dim_B(\pa\O\times[0,1]^{N-1})=\ov\dim_B\pa\O+N-1$.

Since, according to \cite[Theorem 6.21]{lapidusfrank12} (along with Example~\ref{L} and Remark~\ref{entirely}), 
\begin{equation}\label{-rho}
\po(\zeta_{\pa(\O_a)})=\{\rho,-\rho,-2\rho,-3\rho,\dots\},
\end{equation}
where $\rho:=1/(a+1)$, we deduce from Theorem \ref{Axm} that
\begin{equation}\label{OaPC}
\begin{gathered}
\po(\zeta_{\pa(\O_a\times(0,1)^{N-1})})=\po(\zeta_{\pa(\O_a)\times[0,1]^{N-1}})\\
=\{N-1+\rho,N-1-\rho,N-1-2\rho,N-1-3\rho,\dots\},
\end{gathered}
\end{equation}
still with $\rho=1/(a+1)$. Furthermore, all of these complex dimensions are simple.

\begin{remark}
More precisely, it could be that beside $\rho$, which is always a (simple) pole of $\zeta_{\pa\O}$, some of the numbers $-n\rho$ ($n\ge1$) appearing in \eqref{-rho} are not poles of $\zeta_{\pa\O}$ (because the corresponding residue of $\zeta_{\pa\O}$ happens to vanish, for some arithmetic reason connected with the value of $a$). And, hence, similarly, in \eqref{OaPC}.
\end{remark}

Note that if, in Example \ref{combs} just above, $\O=\O_{CS}$ is the Cantor string (i.e., the complement of the classic ternary Cantor set in $[0,1]$), then according to \cite[Equation (1.30)]{lapidusfrank12} and Equation \eqref{OaPC}, we have
\begin{equation}
\dim_{PC}\pa(\O\times(0,1)^{N-1})=\big((N-1)+\log_32\big)+\frac{2\pi}{\log3}\I\Ze,
\end{equation}
which is the special case of \eqref{CmaPC} corresponding to $m:=2$, $a:=1/3$ and $d:=N-1$.
\end{example}

\section[Transcendentally $n$-quasiperiodic sets]{Transcendentally $n$-quasiperiodic sets and their distance zeta functions}\label{quasi0}

The goal of this section is to
describe a construction of some of the simplest
classes of quasiperiodic sets, a notion which we introduce in Definition~\ref{quasiperiodic} below. The main result is obtained in Theorem~\ref{quasi1}.
The construction will be carried out by using a class of generalized Cantor sets depending on two auxiliary parameters. 
We note that, as will be briefly discussed in Subsection \ref{hyperfractal} below, this construction and its natural generalizations will play a key role in future developments of the present higher-dimensional theory of complex dimensions of fractals; see the corresponding discussion in Remark \ref{remr4} and Subsection \ref{hyperfractal} below.


\subsection{Generalized Cantor sets defined by two parameters}\label{cantor_ma}
Let us introduce a class of generalized Cantor sets $C^{(m,a)}$, depending on two parameters. As a special case, we obtain the Cantor sets of the form $C^{(a)}:=C^{(2,a)}$
discussed in Example~\ref{res-cantor}. The classical ternary Cantor set $C^{(1/3)}$ corresponds to the case when $m:=2$ and $a:=1/3$.

\begin{defn}\label{Cma}
The generalized Cantor sets $C^{(m,a)}$ are determined by an integer $m\ge2$ and a positive real number $a$ such that $ma<1$.
In the first step of the analog of Cantor's construction, we start with $m$ equidistant, closed intervals in $[0,1]$ of length $a$, with $m-1$ holes, each of length $(1-ma)/(m-1)$. In the second step, we continue by scaling by the factor $a$ each of the $m$ intervals of length $a$; and so on, ad infinitum.
The  $($two-parameter$)$ {\em generalized Cantor set} $C^{(m,a)}$ is defined as the intersection of the decreasing sequence of compact sets constructed in this way.
\end{defn}

It can be shown that the generalized Cantor sets $C^{(m,a)}$ have the following properties, which extend the ones established for
the sets $C^{(a)}$. 
Apart from the proof of \eqref{zetaCma}, which is easily obtained, the proof of the proposition
is similar to that for the standard Cantor set (see \cite[Equation (1.11)]{lapidusfrank12}), and therefore, we omit it.

\begin{prop}\label{Cmap}
 If $C^{(m,a)}\st\eR$ is the generalized Cantor set introduced in Definition~\ref{Cma}, then
\begin{equation}\label{2.1.1}
D:=\dim_B C^{(m,a)}=D(\zeta_A)=\log_{1/a}m.
\end{equation}
Furthermore, the tube formula associated with $C^{(m,a)}$ is given by
\begin{equation}\label{Cmat}
|C^{(m,a)}_t|=t^{1-D}G(\log t^{-1})
\end{equation}
for all $t\in(0,\frac{1-ma}{2(m-1)})$, where $G=G(\tau)$ is the following nonconstant, positive and bounded periodic function,
with minimal period equal to $T=\log (1/a)$, and defined by 
\begin{equation}\label{Gtau}
G(\tau)=c^{D-1}(ma)^{g\left(\frac{\tau-c}{T}\right)}+2\,c^Dm^{g\left(\frac{\tau-c}{T}\right)}.
\end{equation}
Here, $c=\frac{1-ma}{2(m-1)}$, and $g:\eR\to\eR$ is the $1$-periodic function defined by $g(x)=1-x$ for $x\in(0,1]$.

Moreover, the lower and upper Minkowski contents of $C^{(m,a)}$ are respectively given by
\begin{equation}\label{CmaM}
\begin{aligned}
\M_*^D(C^{(m,a)})&=\min G=\frac1D\left(\frac{2D}{1-D}\right)^{1-D},\\
\M^{*D}(C^{(m,a)})&=\max G=\left(\frac{1-ma}{2(m-1)}\right)^{D-1}\frac{m(1-a)}{m-1}.
\end{aligned}
\end{equation}
Therefore, $C^{(m,a)}$ is Minkowski nondegenerate but is not Minkowski measurable.

Finally, if we assume that $\delta\ge\frac{1-ma}{2(m-1)}$, then, the distance zeta function of $A:=C^{(m,a)}$ is given by
\begin{equation}\label{zetaCma}
\zeta_A(s):=\int_{-\delta}^{1+\delta}d(x,A)^{s-1}\D x=\left(\frac{1-ma}{2(m-1)}\right)^{s-1}\frac{1-ma}{s(1-ma^s)}+\frac{2\delta^s}s.
\end{equation}
As a result, $\zeta_A(s)$ admits a meromorphic continuation to all of $\Ce$, given by the last expression in $(\ref{zetaCma})$. In particular,
\begin{equation}\label{zetasim}
\zeta_A(s)\sim\frac1{1-ma^s},
\end{equation}
and the set of poles of $\zeta_A$ $($in $\Ce)$ and the residue of $\zeta_A$ at $s=D$ are respectively given by
\begin{equation}\label{2.1.6}
\begin{aligned}
\po(\zeta_A)&=(D+\mathbf p{\I}\Ze)\cup\{0\},\\ 
\res(\zeta_A,D)&=\frac{1-ma}{DT}\left(\frac{1-ma}{2(m-1)}\right)^{D-1},
\end{aligned}
\end{equation}
where $\mathbf p:=2\pi/T=2\pi/\log(1/a)$ is the oscillatory period of $C^{(m,a)}$.
Finally, each pole in $\po(\zeta_A)$ is simple.
\end{prop}


\begin{defn}
According to the terminology introduced in \cite{lapidusfrank12}, the value of $\mathbf p=2\pi/\log(1/a)$, appearing in Proposition~\ref{Cmap}, is called
the {\em oscillatory period}\label{osc_period} of the generalized Cantor set $A=C^{(m,a)}$. 

As we see from Equation \eqref{zetaCma} and from the equivalence in \eqref{zetasim}, the set of all complex dimensions of the generalized Cantor set $A=C^{(m,a)}$ and the set of principal complex dimensions of $A$ are given, respectively, by
$$
\po(\zeta_A)=(D+\mathbf p{\I}\Ze)\cup\{0\}
\q\hbox{\rm and}\q\po_c(\zeta_A)=D+\mathbf p{\I}\Ze.
$$
\end{defn}

\subsection{Construction of transcendentally $2$-quasiperiodic sets}\label{qp_sets}
In Example~\ref{transcendent} below, we provide some basic ideas for further definitions and constructions.
The main result of this subsection is obtained in Theorem~\ref{trans}.

\begin{example}\label{transcendent}
Let us define two generalized Cantor sets $A=C^{(a)}:=C^{(2,a)}\st[0,1]$, $a\in(0,1/2)$,
and $B=C^{(3,b)}\st[2,3]$, where $b\in(0,1/3)$. We choose $b$ so that $D:=\log_{1/a} 2=\log_{1/b}3$. We may take, for example, $a=1/3$ and $b=3^{-\log_23}$. Note that we then have $3b=3^{1-\log_23}<1$.
Also, we have
$$
|A_t|=t^{1-D}G_1(\log t^{-1}),\q
|B_t|=t^{1-D}G_2(\log t^{-1}).
$$
The functions $G_1$ and $G_2$ corresponding to $A$ and $B$ are
$T$ and $S$-periodic, respectively, with $T=\log(1/a)=\log3$ and $S=\log(1/b)$. 
Furthermore, the quotient $T/S=\log3/\log(1/b)=\log_32$ is  transcendental,
which is a well-known result going back to F.\ von Lindemann\label{lindemann} and K.\ Weierstrass;\label{weierstrass} see~\cite[p.\ 4]{baker}.
\end{example}

For our later needs, it will be convenient to introduce the following definition, which partly follows \cite{enc}.

\begin{defn}\label{quasip}
We say that a function $G=G(\tau):\eR\to\eR$ is {\em transcendentally $n$-quasiperiodic} if it is of the form $G(\tau)=H(\tau,\dots,\tau)$,
where $H:\eR^n\to\eR$  is a function which is nonconstant and $T_k$-periodic in its $k$-th component, for each $k=1,\dots,n$, and the periods $T_1,\dots, T_n$ are {\em algebraically} independent (that is, linearly independent over the field of algebraic real numbers). The values of $T_i$ are called the {\em quasiperiods of G}. The least positive integer $n$ for which this definition is valid is called the {\em order of quasiperiodicity} of $G$.
\end{defn}

\begin{remark}\label{quasir}
It is possible to define analogously a class of {\em algebraically $n$-quasiperiodic functions}, but we do not study them here; see \cite{fzf}. 
\end{remark}

\begin{example}
If $G(\tau)=G_1(\tau)+G_2(\tau)$, where the functions $G_i$ are nonconstant and $T_i$-periodic (for $i=1,2$), such that $T_1/T_2$ is transcendental, then $G$
is transcendentally $2$-quasiperiodic (in the sense of Definition~\ref{quasip}). In this case and in the notation of Definition~\ref{quasip}, we have $H(\tau_1,\tau_2):=G_1(\tau_1)+G_2(\tau_2)$.
\end{example}

In the sequel, we shall need a classic result due to Gel'fond and Schneider (see \cite{gelfond}), proved independently by these two authors in 1934.
We state it in a form that will be convenient for our purposes.

\begin{theorem}[Gel'fond--Schneider, \cite{gelfond}]\label{gs}
Let $\rho$ be a positive algebraic number different from one, and let $x$ be an irrational algebraic number. Then $\rho^x$
is transcendental.
\end{theorem}

\begin{defn}\label{quasiperiodic}
Given a bounded subset $A\st\eR^N$, we say that a function $G:\eR\to\eR$ {\em is associated with the set A} (or {\em corresponds to $A$}) if $A$ has the following tube formula:
\begin{equation}\label{quasiperiodictf}
|A_t|=t^{N-D}(G(\log t^{-1})+o(1))\textrm{ as }t\to0^+,
\end{equation}
where $0<\liminf_{\tau\to\infty}G(\tau)\le \limsup_{\tau\to\infty}G(\tau)<\ty$. Note that it then follows that $\dim_BA$ exists and is equal to $D$.

In addition, we say that $A$ is a {\em transcendentally $n$-quasiperiodic set} if the corresponding function $G=G(\tau)$ is transcendentally $n$-quasiperiodic.
\end{defn}

Generalizing the idea of Example \ref{transcendent} above, we obtain the following result.

\begin{theorem}\label{trans}
Let $A_1=C^{(m_1,a_1)}\st[0,1]$ and $A_2=C^{(m_2,a_2)}\st[2,3]$ be two generalized Cantor sets $($see Definition~\ref{Cma}$\,)$ such that their box dimensions coincide, with the common value $D\in(0,1)$. 
Let $\{p_1,p_2,\dots,p_k\}$ be the set of all distinct prime factors of $m_1$ and $m_2$, and write
\begin{equation}
m_1=p_1^{\alpha_1}p_2^{\alpha_2}\dots p_k^{\alpha_k},\quad m_2=p_1^{\beta_1}p_2^{\beta_2}\dots p_k^{\beta_k},
\end{equation}
where $\alpha_i,\beta_i\in\eN\cup\{0\}$\label{n_0} for $i=1,\ldots,k$. If the exponent vectors 
\begin{equation}
(\alpha_1,\alpha_2,\dots,\alpha_k)\q\mathrm{and}\q(\beta_1,\beta_2,\dots,\beta_k),
\end{equation}
corresponding to $m_1$ and $m_2$,
are linearly independent over the rationals, then the function $G=G_1+G_2$, associated with $A=A_1\cup A_2$,
is transcendentally $2$-quasiperiodic; that is, the quotient $T_1/T_2$ of the quasiperiods of $G$ $($i.e., of the periods of $G_1$ and $G_2$$)$ is transcendental.

Moreover, we have that
$$
\zeta_{A}(s)\sim \frac1{1-m_1a_1^s}+\frac1{1-m_2a_2^s},\q D(\zeta_{A})=D,\q D_{\rm mer}(\zeta_{A})=-\ty,
$$
and hence, the set $\dim_{PC}A=\po_c(\zeta_A)$ of principal complex dimensions of $A$ coincides with the following nonarithmetic set$:$
$$
\dim_{PC} A=D+\Big(\frac{2\pi}{T_1}\,\Ze\,\cup\,\frac{2\pi}{T_2}\,\Ze\Big){\I}.
$$
Besides $(\dim_{PC}A)\cup\{0\}$, there are no other poles of the distance zeta function $\zeta_{A}$.
In other words, $\po(\zeta_A)=\po_c(\zeta_A)\cup\{0\}$. Furthermore, all of the complex dimensions are simple. 

Finally, exactly the same results hold for the tube zeta function $\tilde\zeta_A$ $($instead of~$\zeta_A$$)$.
\end{theorem}

\medskip

\begin{proof}
First of all, using (\ref{Cmat}), applied to both $A_1$ and $A_2$, we conclude that for all $t\in(0,1/2)$,
$$
|(A_1\cup A_2)_t|=t^{1-D}\left(G_1(\log t^{-1})+G_2(\log t^{-1})\right).
$$
It thus suffices to show that the quotient $T_1/T_2$ of the quasiperiods $T_1$ and $T_2$ of the function $G(\tau):=G_1(\tau)+G_2(\tau)$ is transcendental.

From $D=\log_{1/a_1}m_1=\log_{1/a_2}m_2$ and $T_i=\log m_i$, $i=1,2$, we deduce that $x:=T_1/T_2$ satisfies the equation $(m_2)^x=m_1$. The exponent $x$ cannot be an irrational algebraic number,
since otherwise, by the Gel'fond-Schneider theorem (Theorem \ref{gs}), $(m_2)^x$ would be transcendental. If $x$ were rational, say, $x=b/a$, with $a,b\in\eN$ (note that $x>0$, since
$m_1\ge2$), this would then imply that 
$(m_1)^a=(m_2)^b$; that is, 
$$
p_1^{a\alpha_1}p_2^{a\alpha_2}\dots p_k^{a\alpha_k}=p_1^{b\beta_1}p_2^{b\beta_2}\dots p_k^{b\beta_k}.
$$
Therefore, using the fundamental theorem of arithmetic, we would have 
$$
a(\alpha_1,\alpha_2,\dots,\alpha_k)=b(\beta_1,\beta_2,\dots,\beta_k).
$$ 
However, this is impossible due to the assumption of linear independence over the rationals of the above exponent vectors. 
Consequently, $x$ is transcendental.

The claims about the zeta function $\zeta_{A_1\cup A_2}$ follow from Proposition~\ref{Cmap} applied to both $A_1$ and $A_2$.
Indeed, since $A_1$ and $A_2$ are subsets of two disjoint compact intervals, then $\zeta_A(s)\sim\zeta_{A_1}(s)+\zeta_{A_2}(s)$, and
on the other hand, $\zeta_{A_1}(s)+\zeta_{A_2}(s)\sim (1-m_1a_1^s)^{-1}+(1-m_2a_2^s)^{-1}$, in light of \eqref{zetasim} applied separately to $A_1$ and~$A_2$.
This completes the proof of the theorem.
\end{proof}

\begin{remark}
Theorem~\ref{trans} provides a construction of the set $A=A_1\cup A_2$, such that 
the set $\dim_{PC} A:=\po_c(\zeta_A)$ of principal complex dimensions of $A$ is equal to the union of two 
(discrete) sets of complex dimensions, each of them composed of poles in infinite vertical arithmetic progressions, but with algebraically incommensurable {\em oscillatory quasiperiods}\label{oscqp} $\mathbf p_1=2\pi/T_1$ and $\mathbf p_2=2\pi/T_2$ of $A_1$ and $A_2$,
respectively; that is, such that $\mathbf p_1/\mathbf p_2$ is transcendental.
These oscillatory quasiperiods of $A$ are equal to the oscillatory periods of $A_1$
and $A_2$, respectively.  
\end{remark}


\subsection{Transcendentally $n$-quasiperiodic sets and Baker's theorem}\label{qp_sets_baker}
 The main result of this subsection is stated in Theorem~\ref{quasi1} below, which extends Theorem~\ref{trans} to any integer $n\ge2$
and also provides further helpful information.
In the sequel, we shall need the following important theorem from transcendental number theory, due to Baker \cite[Theorem~2.1]{baker}. 
It represents a nontrivial extension of Theorem~\ref{gs}, due to Gel'fond and Schneider \cite{gelfond}.
Recall that an algebraic number is a complex root of a polynomial with integer coefficients and that the field of algebraic numbers is isomorphic to the algebraic closure of $\Qu$, the field of rational numbers.

\begin{theorem}[{Baker, \cite[Theorem~2.1]{baker}}]
\label{baker0}
Let $n\in\eN$ with $n\geq 2$.
If $m_1,\dots, m_n$ are positive algebraic numbers such that $\log m_1,\dots,\log m_n$ are linearly independent over the rationals,
then 
$$
1,\log m_1,\dots,\log m_n
$$ 
are linearly independent over the field of all algebraic numbers.
\end{theorem}

We now state the main result of this subsection, which can be considered as a fractal set-theoretic interpretation of Baker's theorem. It extends Theorem~\ref{trans} even in the case when $n:=2$.

\begin{theorem}\label{quasi1}
Let $n\in\eN$ with $n\geq 2$.
Assume that $A_i=C^{(m_i,a_i)}$, $i=1,\dots,n$, are generalized Cantor sets $($in the sense of Definition~\ref{Cma}$)$ such that their box dimensions are all equal to a fixed number 
$D\in(0,1)$. Assume that there is a disjoint family of closed unit intervals~$I_1,\dots,I_n$ on the real line,
 such that $A_i\st I_i$ for each $j=1,\dots,n$.
Let $T_i:=\log(1/a_i)$ be the associated periods, and $G_i$ be the corresponding $($nonconstant$)$ $T_i$-periodic functions, for $i=1,\dots, n$.
Let $\{p_j:j=1,\dots,k\}$ be the union of all distinct prime factors  which appear in the integers $m_i$, for $i=1,\dots, n$; that is, $m_i=p_1^{\a_{i1}}\dots p_k^{\a_{ik}}$, where $\a_{ij}\in\eN\cup\{0\}$.

If the exponent vectors $e_i$ of the numbers $m_i$,
\begin{equation}\label{2.1.111/2}
e_i:=(\a_{i1},\dots,\a_{ik}),\quad i=1,\dots,n,
\end{equation}
are linearly independent over the rationals, then the numbers
\begin{equation}\label{Ts}
\frac 1D,T_1,\dots, T_n
\end{equation}
are linearly independent over the field of all algebraic numbers.
It follows that the set $A:=A_1\cup\dots\cup A_n\st\eR$ is transcendentally $n$-quasiperiodic; see Definition~\ref{quasiperiodic}.
Furthermore, in the notation of Definition~\ref{quasiperiodic}, an associated transcendentally $n$-quasiperiodic function $G$ is given by $G:=G_1+\cdots+G_n$.

Moreover, we have that
$$
\zeta_{A}(s)\sim \sum_{i=1}^n\frac1{1-m_ia_i^s},\q D(\zeta_A)=D,\q D_{\rm mer}(\zeta_A)=-\ty,
$$
and hence, the set $\dim_{PC}A=\po_c(\zeta_A)$ of principal complex dimensions of $A$ consists of simple poles and coincides with the following nonarithmetic set$:$
$$
\dim_{PC} A=D+\Big(\bigcup_{i=1}^n\frac{2\pi}{T_i}\,\Ze\Big){\I}.
$$
Besides $(\dim_{PC}A)\cup\{0\}$, there are no other poles of the distance zeta function $\zeta_A$.
That is, $\po(A)=\po_c(A)\cup\{0\}$. Furthermore, all of these complex dimensions are simple.

Finally, exactly the same results hold for the tube zeta function $\tilde\zeta_A$ $($instead of $\zeta_A$$)$.
\end{theorem}

\begin{proof}
As in the proof of Theorem~\ref{trans}, 
using (\ref{Cmat}), applied to each $A_i$, for $i=1,\dots,n$, we see that for all $t>0$ small enough,
$$
|A_t|=t^{1-D}\sum_{i=1}^n G_i(\log t^{-1}),
$$
and for each $i=1,\ldots,n$, $G_i=G_i(\tau)$ is $T_i$-periodic, where $T_i:=\log a_i^{-1}$. We next proceed in three steps:

\bigskip

{\em Step} 1: It is easy to check that the numbers $\log p_j$ (for $j=1,\dots,n$) are rationally independent. Indeed, if we had $\sum_{j=1}^k\g_j\log p_j=0$
for some integers $\g_j$, then $\prod_{j=1}^k p_j^{\g_j}=1$. This implies that $\g_j=0$ for all $j$, since otherwise it would contradict the
fundamental theorem of arithmetic.\footnote{A moment's reflection shows that this argument is valid even if the $\g_j$'s are not a priori all of the same sign.}

\bigskip

{\em Step} 2:
Let us show that $\log m_1,\dots,\log m_n$ are linearly independent over the rationals. Indeed, assume that for $i=1,\ldots,n$, $\mu_i\in\Qu$
are such that $\sum_{i=1}^n\mu_i\log m_i=0$. Then
\begin{equation}
\sum_{i=1}^n\mu_i\sum_{j=1}^k\a_{ij}\log p_j=0.
\end{equation}
Changing the order of summation, we have
\begin{equation}
\sum_{j=1}^k\left(\sum_{i=1}^n  \mu_i\a_{ij}\right)\log p_j=0.
\end{equation}
Since, by Step 1, the numbers $\log p_j$ are rationally independent, we have that for all $j=1,\dots,k$,
$$
\sum_{i=1}^n  \mu_i\a_{ij}=0;
$$
that is, $\sum_{i=1}^n\mu_i e_i=0$, where the $e_i$'s are the exponent vectors given by~\eqref{2.1.111/2}. According to the hypotheses of the theorem, the exponent vectors $e_i$ are rationally independent, and we therefore conclude that $\mu_i=0$ for all $i=1,\dots,n$, as desired.

\bigskip

{\em Step} 3:
Using \cite[Theorem~2.1]{baker}, that is, Theorem~\ref{baker0} above, we conclude that $1,\log m_1,\dots,\log m_n$ are linearly independent over the field of algebraic numbers.
Since $T_i=\frac1D\log m_i$, for $i=1,\dots,n$, it then follows that the numbers listed in (\ref{Ts}) are also linearly independent over the field of algebraic numbers.
Therefore, the function 
$$
G:=G_1+\dots+G_n,\quad G(\tau)=G_1(\tau)+\dots +G_n(\tau),
$$ 
associated with $A$, is transcendentally $n$-quasiperiodic; that is, the set $A$ is transcendentally $n$-quasiperiodic. Note that here, $H(\tau_1,\dots,\tau_n):=G_1(\tau_1)+\dots+ G_n(\tau_n)$, in the notation of Definition~\ref{quasip}.

The last claim, about the distance zeta function $\zeta_A$ and its complex dimensions, now follows from Proposition~\ref{Cmap} applied to each of the bounded sets $A_i$ ($i=1,\ldots,n$).
This concludes the proof of the theorem.
\end{proof}

\begin{remark}\label{remr4}
In Theorem~\ref{quasi1}, we have constructed a class of bounded subsets of the real line possessing an arbitrary prescribed finite number of algebraically incommensurable quasiperiods.
As will be further discussed in Subsection \ref{hyperfractal},
this result is extended in \cite{memoir}, where we construct a bounded subset $A_0$ of the real line which is {\em transcendentally $\ty$-quasiperiodic set}; that is, $A_0$ contains infinitely many algebraically incommensurable quasiperiods.
\end{remark}

In the following proposition, by a {\em quasiperiodic set} we mean a set which has one of the following {\em types of quasiperiodicity}: it is either $n$-transcendentally quasiperiodic (see Definition \ref{quasiperiodic}), or $n$-algebraically quasiperiodic (see Remark \ref{quasir}), for some $n\in \{2,3,\dots\}\cup\{\ty\}$ (the case when $n=\ty$ is treated in [LapRa\v Zu1,3]). We adopt a similar convention for the quasiperiodic functions $G=G(\tau)$ appearing in Definition~\ref{quasip}.

\begin{prop}\label{0L}
 Assume that $A$ is a quasiperiodic set in $\eR^N$ of a given type, with an associated quasiperiodic function $G=G(\tau)$. If $d$ is a positive integer and $L>0$, then the subset $A\times[0,L]^d$ of $\eR^{N+d}$
is also quasiperiodic of the same type, with the associated quasiperiodic function equal to $L^d\cdot G$. 
In particular, if $n\ge2$ is an integer and $A$ is the $n$-quasiperiodic subset of $\eR$ constructed in Theorem \ref{quasi1}, then the subset $A\times[0,L]^d$ of $\eR^{1+d}$ is also $n$-quasiperiodic.
\end{prop}

\begin{proof} Let us first prove the claim for $d=1$.
By assumption, we have that 
\begin{equation}
|A_t|_N=t^{N-D}(G(\log t^{-1})+o(1))\q\mbox{as $t\to0^+$,}
\end{equation}
where $G=G(\tau)$ is a quasiperiodic function; see Equation \eqref{quasiperiodictf}. Much as in Equation \eqref{Aid}, we can write
\begin{equation}
\begin{aligned}
|(A\times[0,L])_t|_{N+1}&=|A_t|_N\cdot L+|(A\times\{0\})_t|_{N+1}\\
&=t^{(N+1)-(D+1)}(L\cdot G(\log t^{-1})+o(1))+|(A\times\{0\})_t|_{N+1}
\end{aligned}
\end{equation}
as $t\to0^+$, where $A\times\{0\}\stq\eR^{N+1}$ (so that $(A\times\{0\})_t$ is the $t$-neighborhood of $A\times\{0\}$ taken in $\eR^{N+1}$).
Since, obviously, $|(A\times\{0\})_t|_{N+1}\le|A_t|_N\cdot 2t$, we have that 
\begin{equation}
\begin{aligned}
|A_t|_{N+1}&\le t^{N+1-D}(G(\log t^{-1})+o(1))=t^{(N+1)-(D+1)}\cdot t(G(\log t^{-1})+o(1))\\
&=t^{(N+1)-(D+1)}\cdot O(t)\q\mbox{ as $t\to0^+$.}
\end{aligned}
\end{equation}
 Therefore, 
\begin{equation}
\begin{aligned}
|(A\times[0,L])_t|_{N+1}&=t^{(N+1)-(D+1)}(L\cdot G(\log t^{-1})+o(1)+O(t))\\
&=t^{(N+1)-(D+1)}(L\cdot G(\log t^{-1})+o(1))\q\mbox{as $t\to0^+$.}
\end{aligned}
\end{equation}
Hence, by Definition \ref{quasiperiodic}, the set $A\times[0,L]$ is quasiperiodic, with the associated quasiperiodic function $L\cdot G$.
This completes the proof of the proposition for $d=1$. The general case is easily obtained by induction on~$d$.
\end{proof}

\subsection{Future applications and extensions: $\ty$-quasiperiodic sets, hyperfractals, and the notion of fractality}\label{hyperfractal}

The results of Section \ref{quasi0} and their various generalizations (and, especially, the construction of $n$-quasiperiodic sets carried out in Subsection \ref{qp_sets_baker} above, once it has been extended to the case when $n=\ty$, as described in Remark \ref{remr4} above) will play a key role in the applications of the higher-dimensional theory of complex dimensions developed in the present paper and in our later work.
This will be so, in particular, in relation to the construction of (transcendentally) $\ty$-quasiperiodic, {\em maximally hyperfractal} sets for which, by definition, the associated fractal zeta functions have a natural boundary 
along the critical line $\{\re s=\ov\dim_BA\}$ and, in fact, have a singularity at every point of that line.
Such sets are as ``fractal'' as possible since, in some sense, they have a continuum of nonreal ``complex dimensions'' (interpreted here as singularities of the fractal zeta functions attached to $A$), in striking contrast with the more usual case where the fractal zeta functions can be meromorphically extended to an open connected neighborhood of the critical line $\{\re s=\ov\dim_BA\}$ and therefore have at most countably nonreal complex dimensions. 

Recall that following \cite[Sections 12.1 and 13.4]{lapidusfrank12} (naturally extended to higher dimensions within the framework of our new theory), 
a bounded subset $A$ of $\eR^N$ is said to be ``fractal'' if its associated fractal zeta function (here, $\zeta_A$ or $\tilde\zeta_A$) has a {\em nonreal} complex dimension or else, if it has a natural boundary along a suitable curve (a screen $\bm S$, in the sense of Subsection \ref{eqzf} above); that is, the tube zeta function $\tilde\zeta_A$ (or, equivalently, the distance zeta function $\zeta_A$ if $\ov\dim_BA<N$) cannot be meromorphically extended beyond~$\bm S$.

We close these comments by noting that throughout Section \ref{quasi0} (with the exception of Proposition \ref{0L}), we have worked with bounded subsets of the real line, $\eR$. However, by using the results of Subsection~\ref{dtx}
(especially, Theorem \ref{Axm}), one can easily obtain corresponding constructions of transcendentally $\ty$-quasiperiodic compact sets $A$ in $\eR^N$ (for any $N\ge1$), with $\ov\dim_BA\in(N-1, N)$. (See also Proposition \ref{0L} at the end of Subsection \ref{qp_sets_baker}.)
Likewise, using Theorem \ref{Axm}, one can construct $\ty$-quasiperiodic maximally hyperfractal compact subsets $A$ of $\eR^N$ (for any $N\ge1$)
such that $\ov\dim_BA\in(N-1,N)$. (Actually, by considering the Cartesian product of the original subset of $\eR$ by $[0,1]^d$, with $0\le d\le N-1$,  one may assume that $\ov\dim_BA\in(d,N)$; the same comment can be made about all of the results obtained in Section \ref{quasi0}.)

Finally, these results can also be applied in a key manner in order to establish the optimality of certain inequalities associated with the meromorphic continuations of the spectral zeta functions of (relative) fractal drums (see \cite[Section 4.3]{fzf} and \cite[Section 6]{brezish}).
More specifically, as is pointed out in \cite{Lap3}, the sharp error estimates obtained in \cite{Lap1} for the eigenvalue counting functions of Dirichlet (or, under appropriate assumptions, Neumann) Laplacians (and more general elliptic operators of order $2m$ with possibly variable coefficients) imply that the corresponding spectral zeta functions admit a meromorphic continuation to a suitable open right half-plane $\{\re s>\d_{\pa\O}\}$, where $\d_{\pa\O}$ 
is the inner (upper) Minkowski dimension of the boundary $\pa\O$.
Our construction of $n$-quasiperiodic sets, as given in Section~\ref{quasi0} of this paper, and extended to $n=\ty$ (as suggested above), enables us to deduce that this inequality is sharp, in general. That is, we construct a bounded open set $\O$ in $\eR^N$ with boundary $A:=\pa\O$, such that the compact set $A\st\eR^N$ is $\ty$-quasiperiodic. It follows that $\s_{\rm mer}$, the abscissa of meromorphic continuation of the corresponding spectral zeta function, satisfies $\s_{\rm mer}\ge\d_{\pa\O}$. For example, for the Dirichlet Laplacian on $\O$, we have $\s_{\rm mer}=\d_{\pa\O}$; i.e., $\{\re s>\d_{\pa\O}\}$ is the largest open right half-plane to which the associated spectral zeta function can be meromorphically continued. In fact, a much stronger statement is true in this case. Namely, the spectral zeta function $\zeta_{\nu}(s)$ has a nonisolated singularity at every point of the vertical line $\{\re s=\d_{\pa\O}\}$.

\bigskip

\noindent{\bf Acknowledgement.} We express our gratitude to the referee for several useful remarks and suggestions.
\bigskip

\section*{Appendix A: Equivalence relation and extended Dirichlet-type integrals}\label{appendix}


One problem with the notion of ``equivalence'' provided in Definition \ref{equ} of Subsection \ref{eqzf} is that, strictly speaking,
it is not an equivalence relation or is not even well defined if (as can be very useful) one wishes to allow $g$ to be a meromorphic function (rather than a DTI) because, a priori, $f$ and $g$ no longer belong to the same class of functions. (Indeed, $f$ is then a Dirichlet-type integral, abbreviated DTI in the sequel, while $g$ is merely assumed to be meromorphic; in particular, the abscissa of convergence of $g$ need not be well defined.) The situation is very analogous, in spirit, to the evaluation of the ``leading part'' ($g=g(s)$, in the present case) of a function ($f=f(s)$, here) in the theory of asymptotic expansions. In that situation, the ``leading part'' $g$ belongs to a scale of typical functions (describing the possible asymptotic behaviors of the function $f$ in the given asymptotic limit).

In our present situation, just as in the theory of asymptotic expansions, formally, the relation $\sim$ is both reflexive and (when it makes sense) transitive. Of course, it is also symmetric when it acts on the same class of functions (for example, DTIs).

However, it is also possible to modify both the definition of $\sim$ and the class of functions on which it acts so that it becomes a true equivalence relation on a single space of functions, namely, the class of extended DTIs. The latter class of (tamed) extended DTIs contains the class of (tamed) DTIs
(hence, all of the functions $f$ we wanted to work with in Definition \ref{equ}) and it also contains (essentially) all of the functions $g$ occurring in practice (when applying Definition \ref{equ}).

By definition, given $r\in(0,1)$, a DTI {\em of base} $r$ is a function of the form 
\begin{equation}
g(s)=\zeta_{E,\f,\mu}(r^{-s}),\tag{A.1}
\end{equation}
where $f(s):=\zeta_{E,\f,\mu}(s)$ is a (standard) DTI defined by
\begin{equation}\tag{A.2}
\zeta_{E,\f,\mu}(s):=\int_E\f(x)^s\D\mu(x).
\end{equation}
(See also Definition \ref{abscissa_f}.) It is then easy to check (using the analogous result for ordinary DTIs) that if $g$ is tamed (i.e., if $f$ is tamed), then
the abscissa of convergence $D(g)$ of $g$ and the half-plane of convergence $\Pi(g):=\{\re s>D(g)\}$ are not well defined. Indeed, note that 
$$
\f(x)^{r^s}=\f(x)^{r^{\re s}(\cos((\log r)\im s)+\I\sin((\log r)\im s)},
$$
so that the open set $V$ of complex numbers $s$ for which $\f(x)^{r^s}$ is Lebesgue integrable on $E$ (typically) consists of countably many connected components, and, hence,
does not have the form of a half-plane. The indicated open set $V$ is analyzed in \cite[Appendix A, Section A.4]{fzf}.
\medskip


\medskip

\noindent{\bf Definition A.1.}\,
An {\em extended Dirichlet-type integral} (an {\em extended} DTI or EDTI, in short) $h=h(s)$ is either of the form
\begin{equation}
h(s):=\rho(s)\zeta_{E,\f,\mu}(s)\tag{A.3}
\end{equation}
or of the form
\begin{equation}
h(s):=\rho(s)\zeta_{E,\f,\mu}(r^{-s}),\q\mbox{for some $r\in(0,1)$,}\tag{A.4}
\end{equation}
where $\rho=\rho(s)$ is a nowhere vanishing entire function and $\zeta_{E,\f,\mu}=\zeta_{E,\f,\mu}(s)$ is a DTI. More generally, $\rho$\label{rho}
can be a holomorphic function which does not have any zeros in the given domain $U\stq\Ce$ under consideration, where $U$ contains the closed half-plane $\{\re s>D(\zeta_{E,\f\mu})\}$.

\medskip

If the extended DTI is of the form (A.3), it is said to be of {\em type} I, and if it is of the form (A.4), it is said to be of {\em type} II (or of type II$_r$ if one wants to keep track of the underlying base $r$). Note that EDTIs of type I include all ordinary DTIs as a special case (by taking $\rho\equiv1$).%

Let us denote by $f(s):=\zeta_{E,\f,\mu}(s)$ the (standard) DTI and by $g(s):=\zeta_{E,\f,\mu}(r^{-s})$ the DTI of base $r$ occurring in (A.4). Then, by definition (and in accordance with Definition A.1), if $h$ is of the form (A.3), its {\em abscissa of convergence} $D(h)$ is given by $D(h):=D(f)$, while if $h$ is of the form (A.4), then $D(h)=+\ty$, that is, $\Pi(h)=\emptyset$.

If the DTI $f(s):=\zeta_{E,\f,\mu}$ is tamed, then the extended DTI $h$ from Definition A.1 (either in (A.3) or in (A.4)) is said to be {\em tamed}.

Finally, given any tamed extended DTI of type I, $h=h(s)$ (as in the first part of Definition A.1), we call 
\begin{equation}\tag{A.5}
\Pi(h):=\{\re s>D(h)\}
\end{equation} 
the {\em half-plane of convergence} of $h$ (which is maximal, in an obvious sense), and (assuming that $D(h)\in\eR$) we call $\{\re s=D(h)\}$ the {\em critical line} of $h$. 
(The tameness condition enables us to show that this half-plane exists and is indeed, maximal.)
Using a classic theorem about the holomorphicity of integrals depending on a parameter, one can show that $h$ is holomorphic on $\Pi(h)$. Hence, $D_{\rm hol}(h)\le D(h)$.

Here, much as in Definition \ref{abscissa_f}, $D(h)$ and $D_{\rm hol}(h)$ denote, respectively, the {\em abscissa of $($absolute$)$ convergence}  and the {\em abscissa of holomorphic continuation} of $h$. Furthermore, if $h$ is given by (A.3) above, we set $D(h)=D(\zeta_{E,\f,\mu})$ and $D_{\rm hol}(h)=D_{\rm hol}(\zeta_{E,\f,\mu})$, where $D(\zeta_{E,\f,\mu})$ and $D_{\rm hol}(\zeta_{E,\f,\mu})$ are defined in Definition \ref{abscissa_f}.

Moreover, if $h=h(s)$ admits a meromorphic continuation to an open connected set $U$ containing the closed half-plane $\{\re s\ge D(h)\}$, we denote (much as was done in Definition \ref{dimc} for the special case of DTIs) by $\po_c(h)$ the set of {\em principal complex dimensions} of $h$; that is, the set of poles of $h$ (in $U$) located on the critical line $\{\re s=D(h)\}$ of $h$:
\begin{equation}
\po_c(h):=\{\o\in U: \mbox{$\o$ is a pole of $h$ and $\re \o=D(h)$}\}.\tag{A.6}
\end{equation}
Clearly, $\po_c(h)$ does not depend on the choice of the domain $U$ satisfying the above condition.

We define similarly $\po(h)=\po(h,U)$, the {\em set of} (visible) {\em complex dimensions} of $h$, relative to $U$:
\begin{equation}
\po(h):=\{\o\in U:\mbox{$\o$ is a pole of $h$}\}.\tag{A.7}
\end{equation}
Clearly, since $h$ is of type I (i.e., is given as in (A.3)), then $\po_c(h)=\po_c(f)$ and $\po(h)=\po(f)$, where $f(s):=\zeta_{E,\f,\mu}(s)$.
\medskip

We can now modify as follows the definition of the ``equivalence relation'' provided in Definition \ref{equ} of Subsection \ref{eqzf}.
\medskip

\noindent{\bf Definition A.2.}\,
Let $h_1$ and $h_2$ be arbitrary tamed, extended DTIs of type I (as in Definition A.1) such that $D(h_1)=D(h_2)=:D$, with $D\in\eR$. Assume that each of $h_1$ and $h_2$ admits a (necessarily unique) meromorphic continuation to an open connected neighborhood $U$ of the closed half-plane $\{\re s\ge D\}$. Then the functions $h_1$ and $h_2$ are said to be {\em equivalent}, and we write $h_1\sim h_2$, if the sets of poles of $h_1$ and $h_2$ on their common vertical line $\{\re s=D\}$ (and the corresponding poles have the same multiplicities): $\po_c(h_1)=\po_c(h_2)$ (where the equality holds between multisets).
\medskip

We conclude this appendix by providing a class of tamed extended DTIs which can be used to determine the ``leading behavior'' of most of the fractal zeta functions used in the present theory.
\bigskip

\noindent{\bf Theorem A.3.}\,
{\em
Let $P\in\Ce[x]$ be a polynomial with complex coefficients. Then $f(s):=1/P(s)$ is a tamed DTI of type I. 

More specifically, if $\deg P=:n\ge1$,
then
\begin{equation}
f(s):=\frac1{P(s)}=\zeta_{E,\f,\mu}(s),\tag{A.8}
\end{equation}
where $E:=[1,+\ty)^n$, $\f(x):=(x_1\cdots x_n)^{-1}$ for all $x\in E$, and
\begin{equation}\tag{A.9}
\mu(\D x_1,\dots,\D x_n):=c\,x_1^{a_1}\frac{\D x_1}{x_1}\dots x_n^{a_n}\frac{\D x_n}{x_n},
\end{equation}
so that its total variation measure $|\mu|$ $($in the sense of local measures$)$ is given by
\begin{equation}\nonumber
|\mu|(\D x_1,\dots,\D x_n):=c\,x_1^{\re a_1}\frac{\D x_1}{x_1}\dots x_n^{\re a_n}\frac{\D x_n}{x_n},
\end{equation}
where $c:=\frac{1}{n!}P^{(n)}(0)$ and $a_1,\dots,a_n$ are the zeros of $P=P(s)$ $($counted according to their multiplicities, so that $P(s)=c\,\Pi_{m=1}^n(s-a_m)$$)$.

Moreover, 
$D(f)=D(\zeta_{E,\f,\mu})\le\max\{\re a_1,\dots,\re a_n\}$.
}

\bigskip

If, in Theorem A.3, we assume that $\deg P=0$, i.e., if $P$ is constant, say $P\equiv1$, then clearly, 
$f(s)=1/P(s)=1=\zeta_{E,\f,\mu}(s)$, where $E:=[1,+\ty)$, $\f(t):=1$ for all $x\in E$, and $\mu:=\d_1$ (the Dirac measure concentrated at $1$). In particular, $f$ is also tamed in this case.

\bigskip

Theorem A.3 is a consequence of the following two facts: 
\medskip

$(i)$ If $f_a(s):=1/(s-a)$, where $a\in\Ce$ is arbitrary, then $f$ is a tamed DTI of type I, given by
\begin{equation}
f(s):=\frac1{s-a}=\zeta_{E,\f_a,\mu_a}(s),\tag{A.10}
\end{equation}
where $E:=[1,+\ty)$, $\f_a(x):=x^{-1}$ for all $x\in E$, and
\begin{equation}
\mu_a(\D x):=x^a\frac{\D x}x;\tag{A.11}
\end{equation}
so that $|\mu_a|(\D x)=x^{\re a}\D x/x$.
Furthermore, $D(f_a)=\re a$. Note that $f_a:=\zeta_{E,\f_a,\mu_a}$ is obviously tamed because $\f_a(x)\le 1$ for all $x\ge1$. An entirely analogous comment can be made about $f=\zeta_{E,\f,\mu}$ in the theorem.
\medskip

$(ii)$ The tensor product of two tamed DTIs is tamed. More specifically, if the DTIs $\zeta_{E,\f,\mu}$ and $\zeta_{F,\psi,\eta}$ are tamed, then their {\em tensor product} is given by the following tamed DTI:
\begin{equation}\tag{A.12}
h(s):=(\zeta_{E,\f,\mu}\otimes\zeta_{F,\psi,\eta})(s)=\zeta_{E\times F,\f\otimes\psi,\mu\otimes\eta}(s),
\end{equation}
where the tensor product $\f\otimes\psi$ is defined by $(\f\otimes\psi)(x,y):=\f(x)\,\psi(y)$ for $(x,y)\in E\times F$ and the tensor product $\mu\otimes\eta$ is the product measure of $\mu$ and $\eta$ (see, e.g., \cite{cohn}). It is easy to check that the DTI $h$ is tamed because (since $\zeta_{E,\f,\mu}$ and $\zeta_{F,\psi,\eta}$ are tamed), we have $0\le\f(x)\le C(\f)$ $|\mu|$-a.e.\ on $E$ and $0\le\psi(x)\le C(\psi)$ $|\eta|$-a.e.\ on $F$, so that $0\le(\f\otimes\psi)(x,y)\le C(\f)\,C(\psi)$ $|(\mu\otimes\eta)|$-a.e.\ on $E\times F$.

Furthermore, $D(h)\le\max\{D(\zeta_{E,\f,\mu}),D(\zeta_{F,\psi,\eta})\}$.

Statement $(i)$ above follows from a direct computation, while statement $(ii)$
is proved by an application of the Fubini--Tonelli theorem (for iterated integrals with respect to positive measures) combined with the inequality (between local positive measures) $|\mu\otimes\eta|\le|\mu|\otimes|\eta|$, followed by an application of the classic Fubini theorem (for iterated integrals with respect to possibly signed or complex measures).

\bigskip

\noindent{\bf Corollary A.4.}\,
{\em The meromorphic function on all of $\Ce$ given by
\begin{equation}\tag{A.13}
h_2(s):=\frac{\rho(s)}{P(r^{-s})},
\end{equation}
where $r\in(0,1)$, $P\in\Ce[x]$ is an arbitrary polynomial with complex coefficients and $\rho$ is a nowhere vanishing entire function,
is a tamed extended DTI of type II.
More specifically, $h_2(s)=\rho(s)\zeta_{E,\f,\mu}(r^{-s})$, where $E$, $\f$ and $\mu$ are given in Theorem A.3 above.}

\bigskip

As was alluded to earlier, in practice, when we apply the (modified) definition of the equivalence relation (see Definition A.2 above),
\begin{equation}
h_1\sim h_2\tag{A.14}
\end{equation}
the meromorphic function $h_1$ is a fractal zeta function (an ordinary DTI of type~I), as well as the function $h_2$ (which gives the ``leading behavior'' of $h_1$, to mimick the terminology of the theory of asymptotic expansions). Hence, the importance of Theorem A.3 in the theory developed in the present paper as well as in its future developments. (See, however, Definition A.6 below and the comments surrounding it.)

We refer the interested reader to \cite[Appendix A]{fzf} for more details about the topics discussed in the present appendix, along with detailed proofs of the main results. 
\medskip

\noindent{\em Remark} A.5.\, The two definitions of the notion of equivalence $\sim$ provided in Definition \ref{equ} and Definition A.2 are compatible in the sense that if, in Definition \ref{equ}, we assume that $f$ (denoted by $h_1$ in Definition A.2) is a DTI (as is the case in Definition \ref{equ}), the meromorphic function $g$ is an extended DTI, then $f\sim g$ in the sense of Definition A.2. Note that the functions $f$ and $g$ of Definition \ref{equ} are denoted by $h_1$ and $h_2$ in Definition A.2. (In particular, $D(g)$ and $\po_c(g)$ are well defined, $D(f)=D(g)$ and $\po_c(f)=\po_c(g)$.) The converse statement clearly holds as well.
\medskip


Finally, it is possible, even likely, that in future applications of the current theory of fractal zeta functions developed in this paper and in our later work, we will need to deal with functions $g$ which are no longer extended DTIs (of  type I), but are meromorphic functions of a suitable kind. In that case, we propose to use the following definition, which is a suitable modification of Definition \ref{equ} and seems well suited to various applications. Strictly speaking, it no longer gives rise to an equivalence relation (since $f$ and $g$ belong to different classes of functions) but in this new sense, the statement $f\overset{\rm asym}{\sim} g$ captures appropriately the idea that ``$f$ is asymptotic to $g$''.
\medskip

\noindent{\bf Definition A.6.}\,
Let $f$ be a tamed EDTI and let $g$ be a meromorphic function, both defined and meromorphic on an open and connected subset $U$ of $\Ce$ containing the closed right half-plane $\{\re s\ge D(f)\}$. Then, the function $f$ is said to be {\em asymptotically equivalent} to $g$, and we write $f\overset{\rm asym}{\sim} g$, if $D(f)=D_{\rm hol}(g)$ (and this common value is a real number),
and the poles of $f$ and $g$ located on the convergence critical line $\{\re s>D(f)\}$ of $f$ (which, by assumption, is also the holomorphy critical line of $g$) coincide and have the same multiplicities. 

More succinctly, and with the obvious notation (compare with Equation \eqref{equ2} in Definition \ref{equ} above), we have
\begin{equation}\tag{A.15}
f\overset{\rm asym}{\sim} g\q\overset{\mbox{\tiny def.}}\Longleftrightarrow\q D(f)=D_{\rm hol}(g)\,\,(\in\eR)\,\,\mbox{and}\,\,\po_c(f)=\po_{c,\rm hol}(g).
\end{equation}
More specifically, we let
\begin{equation}\nonumber
\po_{c,\rm hol}(g):=\{\o\in U:\mbox{$\o$ is a pole of $g$ and $\re\o=D_{\rm hol}(g)$}\}.
\end{equation}%
Furthermore, much as in Definition \ref{equ}, $D(f)$ and $D_{\rm hol}(g)$ are viewed as multisets in Equation~(A.15).
\medskip

\noindent{\em Remark} A.7.\,
Observe that even if $g$ is assumed to be a tamed EDTI,
Definition A.6 may differ from its counterpart used in the rest of this appendix (Definition A.2) or, in particular, in Definition \ref{equ}. Indeed, there are examples of tamed DTIs  $g$ for which $D_{\rm hol}(g)<D(g)$. Therefore, strictly speaking, Definition A.6 does not extend Definition \ref{equ} (or Definition A.2). However, it is stated in the same spirit and seems to often be what is needed, in practice.





\end{document}